\documentclass[11pt]{article}%
\usepackage{cite}
\usepackage{url}
\usepackage{amssymb}
\usepackage{amsfonts}
\usepackage{amsmath}
\usepackage{caption}
\usepackage{subcaption}
\usepackage{capt-of}
\usepackage{relsize}
\usepackage{algorithm}
\usepackage{algpseudocode}
\usepackage{hhline}
\usepackage{wasysym}
\usepackage{textcomp}
\usepackage{amsbsy}
\usepackage{cancel}
\usepackage{graphicx}
\usepackage{enumerate}

\usepackage{amssymb}
\usepackage{amsfonts}
\usepackage{graphicx}
\usepackage{enumerate}
\usepackage{shuffle}

\usepackage{tikz}
\usetikzlibrary{arrows,shapes.geometric,positioning}

\usepackage[hidelinks]{hyperref}

\providecommand{\U}[1]{\protect\rule{.1in}{.1in}}

\allowdisplaybreaks
\makeatletter
\newcommand{\owedge}{ 
	\mathbin{
		\mathchoice
		{\buildowedge{\displaystyle}}
		{\buildowedge{\textstyle}}
		{\buildowedge{\scriptstyle}}
		{\buildowedge{\scriptscriptstyle}}
	} 
}

\newcommand\buildowedge[1]{%
	\begin{tikzpicture}[baseline=(X.base), inner sep=0, outer sep=0]
	\node[draw,circle, line width=0.08mm] (X)  {$#1\wedge$};
	\end{tikzpicture}%
}

\newcommand{\ovee}{ 
	\mathbin{
		\mathchoice
		{\buildovee{\displaystyle}}
		{\buildovee{\textstyle}}
		{\buildovee{\scriptstyle}}
		{\buildovee{\scriptscriptstyle}}
	} 
}

\newcommand\buildovee[1]{%
	\begin{tikzpicture}[baseline=(X.base), inner sep=0, outer sep=0]
	\node[draw,circle, line width=0.08mm] (X)  {$#1\vee$};
	\end{tikzpicture}%
}

\newcommand{\sovee}{ \ {\scriptstyle \ovee} \ }
\newcommand{\bigovee}{\pmb{{\Large \ovee}}}
\newcommand{\sowedge}{\ { \scriptstyle \owedge} \ }
\newcommand{\bigowedge}{\pmb{{\Large \owedge}}}
\makeatother

\newtheorem{theorem}{Theorem}

\newtheorem{definition}[theorem]{Definition}
\newtheorem{example}[theorem]{Example}

\newtheorem{proposition}[theorem]{Proposition}
\newtheorem{remark}[theorem]{Remark}

\newenvironment{proof}[1][Proof]{\noindent\textbf{#1.} }{\ \rule{0.5em}{0.5em}}

\definecolor{darkorange}{rgb}{1.0, 0.55, 0.0}

\allowdisplaybreaks
\newcommand{\z}{{\zeta}}

\newcommand{\hide}[1]{}
\newcommand{\cupdot}{\mathbin{\mathaccent\cdot\cup}}
\newcommand{\bigcupdot}{\bigcup\mkern-12.5mu\cdot\mkern6mu}
\providecommand{\keywords}[1]{\textbf{Keywords:} #1}
\begin{document}

\title{Fuzzy propositional configuration logics}
\author{Paulina Paraponiari\thanks{\protect\includegraphics[height=0.3cm]{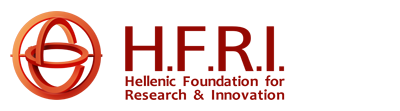}The research work was supported by the Hellenic Foundation for Research and Innovation (HFRI) under the HFRI PhD Fellowship grant (Fellowship Number: 1200).}  \\Department of Mathematics\\Aristotle University of Thessaloniki\\54124 Thessaloniki, Greece\\parapavl@math.auth.gr}
\date{}
\maketitle

\begin{abstract}
	We introduce and investigate a weighted propositional configuration logic over De Morgan algebras. This logic is able to describe software architectures with quantitative features especially the uncertainty of the interactions that occur in the architecture. We deal with the equivalence problem of formulas in our logic by showing that every formula can be written in a specific form. Surprisingly, there are formulas which are equivalent only over specific De Morgan algebras. We provide examples of formulas in our logic which describe well-known software architectures equipped with quantitative features such as the uncertainty and reliability of their interactions.
\end{abstract}


\section{Introduction}

Uncertainty is inevitable in software architecture \cite{uncerta:risk}. Software architectures are increasingly composed of many components such as workload and servers. Computations between the components run in environments in which resources may have radical variability \cite{soft:uncertain:world}. For instance, software architects may be uncertain about the cost and performance impact of a proposed software architecture. They may be aware of the cost and performance of the interactions in the architecture. However, there may be undesirable outcomes such as failure of a component to interact and complete its task \cite{uncerta:risk}. Uncertainty may affect functional and non-functional architecture requirements \cite{relax}. Hence, it is necessary to consider uncertainty as a basic quantitative characteristic in software architectures. So far, the existing architecture decision-making approaches do not provide a quantitative method of dealing with uncertainty \cite{Deali:uncert}. The motivation of our work is to formally describe and compare software architectures with quantitative features such as the uncertainty. For this, we consider that fuzzy logics rely on the idea that truth comes in degrees. Hence, they constitute a suitable tool in order to deal with uncertainty. Moreover, recently the authors in \cite{fuzzy:iot} used fuzzy logic on IoT devices for assisting the blind people for their safe movement. This is a strong indication for the possible future applications of fuzzy logic.

In this paper we extend the work of \cite{Ka:Pa,Pa_Rah_1,Pa_Rah} by introducing and investigating the fuzzy PCL (fPCL for short) over De Morgan algebras. This work is motivated as follows. In \cite{Pa_Rah_1,Pa_Rah} we introduced the weighted PCL over commutative semirings (wPCL for short). This logic serves as a specification language for the study of software architectures with quantitative features such as the maximum cost of an architecture or the maximum priority of the involvement of a component. Then in \cite{Ka:Pa}, we introduced the weighted PCL over product valuation monoids (w$_{\text{pvm}}$PCL for short)  which serves as a specification language for software architectures with quantitative features such as the average of all interactions' costs of the architecture, and the maximum cost among all costs occurring most frequently within a specific number of components. Those features are not covered in \cite{Pa_Rah_1,Pa_Rah}. The aforementioned works are not able to model the uncertainty that occurs between the interactions in the architecture. In this paper we deal with this problem by introducing and investigating the fuzzy PCL (fPCL for short) which is a weighted PCL over De Morgan algebras. 

The contributions of our work are the following.
\begin{enumerate}[$\bullet$]
	\item We introduce the syntax and semantics of fPCL. The semantics of fPCL formulas are series with values in the De Morgan algebra. This logic is able to describe software architecture with quantitative features such as the uncertainty. Moreover, we are able to compute the weight of an architecture even when unwanted components participate. This is possible since De Morgan algebras are equipped with a complement mapping whereas the algebraic structures in \cite{Ka:Pa,Pa_Rah_1,Pa_Rah} are not.
	\item In the sequel, we construct fPCL formulas which describe the Peer-to-Peer architecture and the Master/Slave architecture for finitely number of components. 
	\item  Lastly, we deal with the decidability of equivalence of fPCL formulas. For this, we examine the existence of a normal form. We show that the construction of the normal form of a fPCL formula depends on the properties of the De Morgan algebra. Hence, there may be fPCL formulas which have the same normal form over the fuzzy algebra but different ones over the Boolean algebra. In other words, two fPCL formulas can be equivalent over the fuzzy algebra but not over the Boolean algebra. We give examples to show our point. In our paper, we prove that for every fPCL formula over a set of ports and a Kleene algebra we can effectively construct an equivalent one in normal form. We note that this construction can be easily adapted for fPCL formulas over a Boolean algebra. We conclude that two fPCL formulas are equivalent over a De Morgan algebra if they have the same normal form considering the properties of the aforementioned De Morgan algebra. For this, we give an algorithm which is able to decide the equivalence of two fPCL formulas in normal form, in polynomial time.  
	
\end{enumerate}

\section{Related Work}
Existing work has investigated the formal description of the qualitative and quantitative properties of software architecture. In particular, the authors in \cite{Ma:Co} introduced the propositional configuration logic (PCL for short) which was proved sufficient to describe the qualitative properties of software architectures. Later in  \cite{Pa_Rah_1,Pa_Rah}, we introduced and investigated a weighted PCL (wPCL for short) over a commutative semiring which serves as a specification language for the study of software architectures with quantitative features such as the maximum cost of an architecture or the maximum priority of a component. We proved that the equivalence problem of wPCL formulas is decidable. In \cite{Ka:Pa} we extended the work of \cite{Pa_Rah_1,Pa_Rah} by introducing and investigating weighted PCL over product valuation monoids (w$_\text{pvm}$PCL for short). This logic is proved to be sufficient to serve as a specification language for software architectures with quantitative properties, such as the average of all interactions' costs of the architecture and the maximum cost among all costs occurring most frequently within a specific number of components in an architecture. However, the aforementioned works do not cover quantitative properties such us the uncertainty and reliability of an architecture. 

The authors in \cite{stoch} address the problem of evaluating the system reliability as a stochastic property of software architectural models in the presence of uncertainty. Also, the authors in \cite{Fram:Unc} develop a conceptual framework for the management of uncertainty in software architecture in order to reduce its impact during the system's life cycle. However, the aforementioned works are lack of formality of the architecture description, which is crucial since non-formal systems can be unreliable at some point.

\section{Preliminaries}

\subsection{Lattices}
Let $K$ be a nonempty set, and $\leq$ a binary relation over $K$ which is reflexive, antisymmetric, and transitive. Then $\leq$ is called a partial order  and the pair  $(K, \leq)$  a partially ordered set (poset for short). If the partial order $\leq $ is understood, then we shall denote the poset $(K,\leq)$ simply by $K$. For $k, k' \in K$ we denote by $k \vee k'$ (resp. $k \wedge k'$) the least upper bound or supremum (resp. the greatest lower bound or infimum) of $k$ and $k'$ if it exists in $K$. 

A poset $K$ is called a \emph{lattice} if  $k\vee k'$ and $k\wedge k'$ exist in $K$ for every $k,k'\in K$. A lattice $K$ is called \emph{distributive} if $ k\wedge(k'\vee k'')=(k\wedge k')\vee(k\wedge k'')$ and $(k \vee k')\wedge k''=(k \wedge k'')\vee(k'\wedge k'')$ for every $k,k',k'' \in K$. Moreover, the absorption laws $k \vee \left( k \wedge k^\prime  \right) = k$ and $ k \wedge \left( k \vee k^\prime  \right) = k$ hold for every $k, k^\prime \in K.$ A poset $K$ is called \emph{bounded} if there are two elements $0,1 \in K$ such that $0 \leq k \leq 1$ for every  $k \in K$.

A \emph{De Morgan algebra} is denoted by $(K,\leq,^{-})$, where $K$ is a bounded distributed lattice (bdl for short) with complement mapping $^- : K \rightarrow K$ which satisfies involution and the De Morgan laws $\overline{\overline{k}}=k$, $\overline{k \vee k'}=\overline{k} \wedge \overline{k'},$ and $\overline{k \wedge k'} = \overline{k} \vee \overline{k'} $ for every $k, k' \in K$. A known De Morgan algebra is the structure $([0,1], \leq, ^-)$ where $\leq$ is the usual order on real numbers and the complement mapping is defined by $\overline{k}=1-k$ for every $k \in [0,1]$. 

The authors in \cite{Dr:Mu,Ra:Fu} show that a semiring $(K, +, \cdot, 0,1)$ equipped with a complement mapping $^-$, which is a monoid morphism from $(K,+,0)$ to $(K,\cdot,1)$ and $\overline{\overline{k}}=k$ for every $k \in K$, is a De Morgan algebra $(K, \leq, ^-)$. The relation $\leq$ is defined as follows: $k \leq k'$ iff $k+k'=k'$. On the other hand, a De Morgan algebra $(K, \leq, ^-)$ induces a semiring $(K, \vee, \wedge, 0, 1)$ with a complement mapping $^-$. In the following, we denote a De Morgan algebra by $(K, \vee, \wedge , 0,1, ^-)$. Moreover, a \emph{Kleene algebra} is a De Morgan algebra that satisfies $k_1 \wedge \overline{k_1} \leq k_2 \vee \overline{k_2}$, or equivalently, $(k_1 \wedge \overline{k_1}) \wedge (k_2 \vee \overline{k_2}) = (k_1\wedge \overline{k_1})$ for every $k_1, k_2 \in K$. A \emph{Boolean algebra} is a Kleene algebra that satisfies $k\wedge \overline{k} =0$ and $k\vee \overline{k} = 1$ for every $k\in K.$ In the following we present the most well-known De Morgan algebras. We refer the reader to \cite{Wa:Ge,Mo:Fu} for further examples of De Morgan algebras.

\begin{figure}[t]
	\begin{subfigure}[b]{0.4\textwidth}
		\centering
		\resizebox{.5\textwidth}{!}{$\begin{array}{|c||c| c|c|}
			\hline
			\vee & 0 & 1 & u \\
			\hhline{|=||=|=|=|} 0 & 0 & 1 & u \\ 
			\hline 1 & 1 & 1 & 1\\
			\hline u & u & 1 &u \\ \hline
			\end{array} \hspace*{0.4cm} \begin{array}{|c||c| c|c|}
			\hline
			\wedge & 0 & 1 & u \\
			\hhline{|=||=|=|=|} 0 & 0 & 0 & 0 \\ 
			\hline 1 & 0 & 1 & u\\
			\hline u & 0 & u &u \\ \hline
			\end{array}$}
		\caption{Three element Kleene algebra}
		\label{three_element}
	\end{subfigure}
	\hfil
	\begin{subfigure}[b]{0.5\textwidth}
		\centering
		\resizebox{.5\textwidth}{!}{$\begin{array}{|c||c| c|c|c|}
			\hline
			\vee & 0 & 1 & u & w  \\
			\hhline{|=||=|=|=|=|} 0 & 0 &1 &u & w \\ \hline 
			1 & 1 & 1& 1& 1\\ \hline 
			u &  u & 1&  u& 1 \\ \hline
			w & w & 1 & 1 & w \\ \hline  
			\end{array}   \hspace*{0.4cm}  \begin{array}{|c||c| c|c|c|}
			\hline
			\wedge  & 0 & 1 & u & w  \\
			\hhline{|=||=|=|=|=|} 0 & 0 &0  & 0& 0 \\ \hline 
			1 & 0 &1 &u & w \\ \hline 
			u & 0  &u & u & 0 \\ \hline
			w & 0 &w  & 0 & w \\ \hline  
			\end{array}$}
		\caption{Four element algebra}
		\label{four_element}
	\end{subfigure}
	\caption{Operators of De Morgan algebras}
	\label{op_de_morg}
\end{figure}

\begin{enumerate}[$\bullet$]
	\item The two element Boolean algebra $\textbf{2} = \left( \{0,1\}, \vee, \wedge, 0,1, ^-  \right)$, where $\overline{0}=1$ and $\overline{1}=0$. 
	\item The three element Kleene algebra $\textbf{3}= \left( \{0,u,1\}, \vee, \wedge, 0,1, ^-  \right)$, where $\overline{0}=1$, $\overline{1}=0$, $\overline{u}=u$. The operators $ \vee, \wedge$ are shown in Figure \ref{three_element}. 
	\item The four element algebra $\textbf{4}=\left( \{0,u,w,1\}, \vee, \wedge, 0,1, ^- \right)$, where $\overline{u} = u$, $\overline{w} = w$, $u\vee w = 1$ and $u\wedge w = 0$. The operators $ \vee$ and $\wedge$ are shown in Figure \ref{four_element}.
	\item The fuzzy algebra $\textbf{F}=\left(  [0,1], \max, \min, 0,1, ^-  \right)$, where for every $k\in [0,1]$ the complement mapping is defined by $\overline{k} = 1-k.$ This algebra is a Kleene algebra. To see this, let $k, k^\prime \in [0,1]$ and note that $\min \{    \min\{ k, \overline{k} \}, \max\{ k^\prime, \overline{k^\prime} \}       \} = \min \{ k, \overline{k} \}.$
\end{enumerate}
\begin{quotation}
	\emph{Throughout the paper, $\mathbf{3}$ and $K_\mathbf{3}$ will denote respectively, the three element Kleene algebra and a De Morgan algebra which is a Kleene algebra. Also, by $\mathbf{2}$ and $\textbf{B}$ we will denote respectively, the two element Boolean algebra and a De Morgan algebra which is a Boolean algebra. By $K$ we will denote an arbitrary De Morgan algebra.} 
\end{quotation}

Lastly, consider $K$ be a De Morgan algebra and $Q$ a set. A \emph{formal series} (or simply \emph{series}) \emph{over}
$Q$ \emph{and} $K$ is a mapping $s:Q\rightarrow K$. We denote by  $K\left\langle \left\langle Q
\right\rangle \right\rangle $ the class of all series over $Q$ and $K$.

\section{Fuzzy Propositional Interaction Logic}
In this section we introduce a quantitative version of PIL where the weights are taken in the De Morgan algebra $K$. Since De Morgan algebras and more generally bdl's found applications in fuzzy theory, we call our weighted PIL a fuzzy PIL. 
\begin{definition}
	The syntax of formulas of \emph{fuzzy PIL} (\emph{fPIL} for short) over $P$ and $K$ is given by the grammar:
	$$\varphi::= true \mid p \mid \  ! \varphi \mid \varphi \sovee \varphi  $$
	where $p \in P$ and the operators $!, \sovee  $ denote the fuzzy negation and the fuzzy disjunction, respectively, among \emph{fPIL} formulas. 
\end{definition}

The fuzzy conjunction operator among fPIL formulas $\sowedge$ is defined by $\varphi_1 \sowedge \varphi_2 : = \ ! (! \varphi_1   \sovee  ! \varphi_2).$

For the semantics of fPIL formulas over $P$ and $K$ we introduce the notion of a $K$-fuzzy interaction. For this we need to recall the $K$-fuzzy sets from \cite{l:fuzzy}. A $K$-fuzzy set $S$ on a non empty set $X$ is a function $S: X\to K$. A \emph{$K$-fuzzy interaction} $\alpha$ on $P$ is a $K$-fuzzy set on $P$ with the restriction that $\alpha(p) \neq 0$ for at least one port $p\in P$. We denote by $fI(P,K)$ the set of $K$-fuzzy interactions $\alpha$ on $P$ and by $fPIL(K,P)$ the set of all fPIL formulas over $P$ and $K$. We interpret fPIL formulas over $P$ and $K$ as series in $K \left\langle \left \langle fI(P,K) \right\rangle \right \rangle$.

\begin{definition}
	Let $\varphi \in fPIL(K,P)$. The semantics of $\varphi$ is a series $\left\Vert \varphi \right\Vert \in K \left\langle \left\langle fI(P,K) \right\rangle \right\rangle$. For every $K$-fuzzy interaction $\alpha\in fI(P,K)$ the value $\left\Vert \varphi \right\Vert (a)$ is defined inductively on the structure of $\varphi$ as follows:
	\begin{enumerate}[$\bullet$]
		\item $\left\Vert true\right\Vert (a) = 1$,
		\item $\left\Vert p\right\Vert (a) = a(p)$,
		\item $\left\Vert ! \varphi \right\Vert (a) = \overline{\left\Vert \varphi \right\Vert (a)}$,
		\item $\left\Vert \varphi_1 \sovee \varphi_2 \right\Vert (a) = \left\Vert \varphi_1 \right\Vert (a) \vee \left\Vert \varphi_2 \right\Vert (a) $.
	\end{enumerate}
\end{definition}

Trivially, we get $\left\Vert\varphi_1 \sowedge \varphi_2 \right\Vert(a)=\left\Vert\varphi_1\right\Vert(a)\wedge \left\Vert\varphi_2\right\Vert(a)$ for every  $\alpha\in fI(P,K)$. Moreover, we define the fPIL formula $! true := false$ and it is valid that $\left\Vert  false \right\Vert (\alpha) = 0$ for every $\alpha \in fI(P,K).$ 

Next, we define the equivalence relation among fPIL formulas. For this, we consider that De Morgan algebras such as the Kleene and the Boolean algebra satisfy some extra properties except from the ones that are valid to every De Morgan algebra by its definition. 

\begin{definition}\label{fpil_equiv}
	Two \emph{fPIL} formulas $\varphi_1, \varphi_2$ over $P$ and a concrete De Morgan algebra $K_{con}$ are called $K_{con}$-equivalent, and we write $\varphi_1 \ \dot{\equiv}_{K_{con}}  \ \varphi_2$, whenever $\left\Vert \varphi_1 \right\Vert (\alpha)=\left\Vert \varphi_2\right\Vert(\alpha)$ for every $\alpha \in fI(P,K_{con}).$ 
	
	Two \emph{fPIL} formulas $\varphi_1, \varphi_2$ over $P$ and an arbitrary De Morgan algebra $K$ are called simply equivalent, and we write $\varphi_1\ \dot{\equiv} \ \varphi_2 $, whenever $\left\Vert \varphi_1 \right\Vert (\alpha)=\left\Vert \varphi_2\right\Vert(\alpha)$ for every $\alpha \in fI(P,K).$
\end{definition}

Let $P=\{ p,q,r \}$ be a set of ports. Following the previous definition and by the properties of De Morgan algebras, we prove that $p\sowedge !p \sowedge \left( q\sovee r \right) \ \dot{\equiv}_{\mathbf{2}} \ false$ and $p\sowedge !p \sowedge \left( q\sovee r \right) \ \dot{\equiv} \ \left(p\sowedge !p\sowedge q  \right) \sovee \left(p\sowedge !p\sowedge r  \right)  $.  We proceed with some properties of our fPIL formulas.

\begin{proposition}\label{neg_i_oplus}
	Let $\varphi_1, \varphi_2 $ be fPIL formulas over $P$ and $K$. Then 
	\[  ! \left( \varphi_1 \sovee \varphi_2 \right) \ \dot{\equiv} \ (!\varphi_1) \sowedge \left( ! \varphi_2 \right).  \]
\end{proposition}

\begin{proof}
	Let $\alpha \in fI(P,K)$. Then 
	\begin{align*}
	\left\Vert ! \left( \varphi_1 \sovee \varphi_2  \right) \right\Vert(\alpha) & = \overline{\left\Vert   \varphi_1 \sovee \varphi_2  \right\Vert(\alpha)} \\ & = \overline{  \left\Vert \varphi_1 \right\Vert(\alpha) \vee \left\Vert \varphi_2 \right\Vert(\alpha) } \\ & = \overline{  \left\Vert \varphi_1 \right\Vert(\alpha)  } \wedge \overline{\left\Vert \varphi_2 \right\Vert(\alpha)} \\ & = \left\Vert ! \varphi_1\right\Vert(\alpha) \wedge \left\Vert !\varphi_2\right\Vert(\alpha) \\ & = \left\Vert \left( ! \varphi_1  \right) \sowedge \left( ! \varphi_2  \right) \right\Vert(\alpha).
	\end{align*}   \end{proof}

\begin{proposition}\label{pil_true_false}
	Let $\varphi$ be a \emph{fPIL} formula over $P$ and $K$. Then the following hold:
	\begin{flushleft}
		\begin{tabular}{l l l}
			$ (1)$ $  \varphi \sovee true \ \dot{\equiv} \ true,$ & \hspace*{1cm} $ (3) $ $\varphi \sowedge true \ \dot{\equiv} \ \varphi$, & \hspace*{1cm} $ (5)$ $ !!\varphi \ \dot{\equiv} \ \varphi.  $  \\ $(2)$ $ \varphi \sovee false \ \dot{\equiv}  \ \varphi,$ & \hspace*{1cm} $(4) $ $ \varphi \sowedge false \ \dot{\equiv}  \ false ,$ & 
		\end{tabular}
	\end{flushleft}
\end{proposition}

\begin{proof}
	The proofs are straightforward.
\end{proof}

\begin{proposition}\label{fpil_associa}
	The operators $\sowedge$ and $\sovee$ of the \emph{fPIL} are associative.
\end{proposition}

\begin{proof}
	The proposition holds since the operators $\wedge$ and $\vee$ are associative.	
\end{proof}

\begin{proposition}\label{otimes_over_oplus_i}
	Let $\varphi, \varphi_1, \varphi_2 \in fPIL(K,P)$. Then 
	\[  \varphi \sowedge \left(  \varphi_1 \sovee \varphi_2 \right) \ \dot{\equiv} \ \left( \varphi \sowedge  \varphi_1\right) \sovee \left( \varphi \sowedge \varphi_2 \right). \]
\end{proposition}

\begin{proof}
	Since $\wedge$ distributes over $\vee$ we get the proposition. 
\end{proof}

Next, we give the absorption and idempotent laws among fPIL formulas.

\begin{proposition}\label{absorpt_pil}
	Let $\varphi, \varphi^\prime \in fPIL(K,P)$. Then 
	
	\begin{tabular}{l l l l}
		$(1)$ & $\varphi \sowedge \left( \varphi \sovee \varphi^\prime \right) \  \dot{\equiv} \ \varphi $ & \hspace*{1cm} $(3)$ & $\varphi \sovee \varphi \ \dot{\equiv} \ \varphi $. \\[0.1cm] $(2)$ & $\varphi \sovee \left( \varphi \sowedge \varphi^\prime \right) \ \dot{\equiv} \ \varphi $. & \hspace*{1cm} $(4)$ & $\varphi \sowedge \varphi \ \dot{\equiv} \ \varphi.$
	\end{tabular}
	
\end{proposition}

\begin{proof}
	For the proof of (1) and (2) we apply the absorption laws of De Morgan algebras. The other are valid since $\wedge $ and $\vee$ are idempotent.  
\end{proof}

\section{Fuzzy Propositional Configuration Logic}
In this section we introduce and investigate the fuzzy PCL over $P$ and $K$.

\begin{definition}
	The syntax of formulas of \emph{fuzzy  PCL} (\emph{fPCL} for short) \emph{over} $P$ \emph{and} $K$ is given by the grammar:
	$$ \zeta :: = \varphi \mid  \neg \zeta \mid \zeta \oplus \zeta \mid  \zeta \uplus \zeta   $$
	where  $\varphi$ is a  \emph{fPIL} formula over $P$ and $K$, $\neg$, $\oplus$ and $\uplus$ denote the fuzzy negation, the fuzzy disjunction and the fuzzy coalescing operator, respectively. 
\end{definition}

Let $\zeta, \zeta^\prime $ be fPCL formulas over $P$ and $K$. The fuzzy conjunction operator among $\zeta$ and $\zeta^\prime$ and the closure operator of $\zeta$ are defined, respectively, as follows:
\begin{flushleft}
	$\begin{array}{l l}
	(1) \ \zeta \otimes \zeta^\prime : = \neg (\neg \zeta \oplus \neg \zeta^\prime), & \hspace*{1cm} (2) \   \sim \zeta := \zeta \uplus true.
	\end{array}$
\end{flushleft}

Next, we denote by $fC(P,K)$ the set of nonempty sets of $K$-fuzzy interactions in $fI(P,K)$, and by $fPCL(K,P)$ the set of fPCL formulas over $P$ and $K$. We define the semantics of fPCL formulas over $P$ and $K$ as series in $K \left\langle\left\langle fC(P,K) \right\rangle\right\rangle $.  

\begin{definition}
	Let $\zeta$ be a \emph{fPCL} formula over $P$ and $K$. The semantics of $\zeta$ is a series $\left\Vert \zeta \right\Vert \in K \left\langle\left\langle fC(P,K) \right\rangle\right\rangle$. For every set $\gamma \in fC(P,K)$ the value  $\left\Vert \zeta \right\Vert (\gamma)$ is defined inductively on the structure of $\z$ as follows:
	\begin{enumerate}[$\bullet$]
		\item $\left\Vert \varphi \right\Vert (\gamma) = \underset{\alpha\in\gamma}{\bigwedge} \left\Vert \varphi \right \Vert(a)$,
		\item $\left\Vert \neg \zeta \right\Vert (\gamma) = \overline{\left\Vert \zeta \right\Vert (\gamma)}$,
		\item $\left\Vert \zeta_1 \oplus \zeta_2 \right\Vert (\gamma) = \left\Vert \zeta_1 \right\Vert (\gamma) \vee \left\Vert \zeta_2 \right\Vert (\gamma), $
		\item $\left\Vert \zeta_1 \uplus \zeta_2\right\Vert (\gamma) = \underset{\gamma=\gamma_1 \cup \gamma_2}{\bigvee}  \left(  \left\Vert \zeta_1 \right\Vert (\gamma_1) \wedge \left\Vert \zeta_2 \right\Vert (\gamma_2) \right)$.
	\end{enumerate}
	
\end{definition}

\noindent It is easy to prove that $\left\Vert true \right\Vert (\gamma) = 1$ and $\left\Vert false \right\Vert (\gamma) = 0$ for every $\gamma \in fC(P,K). $

\begin{definition}\label{fpcl_equiv}
	Two \emph{fPCL} formulas $\zeta_1, \zeta_2$ over $P$ and a concrete De Morgan algebra $K_{con}$ are called $K_{con}$-equivalent, and we write $\zeta_1\equiv_{K_{con}} \zeta_2$, whenever $\left\Vert \zeta_1 \right\Vert (\gamma)=\left\Vert \zeta_2\right\Vert(\gamma)$ for every $\gamma \in fC(P,K_{con}).$ 
	
	Two \emph{fPCL} formulas $\zeta_1, \zeta_2$ over $P$ and an arbitrary De Morgan algebra $K$ are called simply equivalent, and we write $\zeta_1\equiv \zeta_2 $, whenever $\left\Vert \zeta_1 \right\Vert (\gamma)=\left\Vert \zeta_2\right\Vert(\gamma)$ for every $\gamma \in fC(P,K).$ 
\end{definition}

In the following, we examine the relation between the fPIL and fPCL operators on fPIL formulas. Firstly, we show that the application of negation operators $!$ and $\neg$ on a fPIL formula derive in general non equivalent fPCL formulas. Indeed, let $p\in P$ and $\gamma =\{a_1,a_2\} \in fC(P,K)$. Then we have
$$
\Vert ! p \Vert (\gamma) = \bigwedge_{a\in \gamma} \Vert ! p \Vert(a) = \overline{\Vert  p \Vert (a_1)} \wedge  \overline{\Vert  p \Vert (a_2)}  = \overline{a_1(p)} \wedge \overline{a_2(p)}
$$
and
$$
\Vert\neg p \Vert (\gamma) = \overline{\Vert p \Vert(\gamma)}  = \overline{\bigwedge_{a\in \gamma} \Vert  p \Vert(a)}  = \bigvee_{a\in \gamma} \overline{ \Vert  p \Vert(a)} = \overline{\Vert  p \Vert (a_1)} \vee  \overline{\Vert  p \Vert (a_2)}  = \overline{a_1(p)} \vee \overline{a_2(p)}$$
which implies that $ ! p \not \equiv \neg p$.

Similarly, we show that in general $\varphi \sovee \varphi^\prime \not \equiv \varphi \oplus \varphi^\prime$ where $\varphi, \varphi^\prime $ are fPIL formulas. For this, let $P$ be a set of ports, $\varphi=p\in P$ and $\varphi^\prime =p^\prime \in P$, where $p\not = p^\prime$. If $\gamma = \{ \alpha_1, \alpha_2 \}\in fC(P,K)$, then we get $\left\Vert p \sovee p' \right\Vert (\gamma) \not = \left\Vert p \oplus p' \right\Vert (\gamma) $ and so $p \sovee p' \not \equiv p \oplus p'. $ 

However, as we show in the next proposition, the application of the operators  $\sowedge$ and $\otimes$ on fPIL formulas produce equivalent fPCL formulas.

\begin{proposition}\label{otimes_i}
	Let $\varphi_1 , \varphi_2$ be \emph{fPIL} formulas over $P$ and $K$. Then
	\[  \varphi_1 \sowedge \varphi_2 \equiv \varphi_1 \otimes \varphi_2. \]
\end{proposition}

\begin{proof}
	For every  $\gamma \in fC(P,K)$ we compute
	\begin{align*}
	\left\Vert \varphi_1 \sowedge \varphi_2 \right\Vert (\gamma) & = \bigwedge_{\alpha\in \gamma} \left\Vert \varphi_1 \sowedge \varphi_2 \right\Vert (\alpha)  \\ & = \bigwedge_{\alpha\in \gamma} \left(\left\Vert \varphi_1\right\Vert(\alpha) \wedge\left\Vert \varphi_2 \right\Vert (\alpha) \right)  \\ & = \left(\bigwedge_{\alpha\in \gamma} \left\Vert \varphi_1\right\Vert(\alpha)\right) \wedge \left( \bigwedge_{\alpha\in \gamma}  \left\Vert \varphi_2 \right\Vert (\alpha) \right) \\ & = \left\Vert \varphi_1\right\Vert(\gamma) \wedge \left\Vert \varphi_2\right\Vert(\gamma) \\ & = \overline{  \overline{\left\Vert \varphi_1 \right\Vert(\gamma)} \vee \overline{\left\Vert \varphi_2 \right\Vert(\gamma)} } \\ & = \left\Vert \neg \left( \neg \varphi_1 \oplus \neg \varphi_2  \right) \right\Vert (\gamma) \\ & = \left\Vert  \varphi_1 \otimes \varphi_2  \right\Vert (\gamma),
	\end{align*}
	where the third equality holds by the commutativity and associativity of $\wedge$.
\end{proof}

In the sequel, we prove several properties of our fPCL formulas. 

\begin{proposition}
	The \emph{fPCL} operators $\oplus, \otimes $ and $\uplus $ are associative and commutative.
\end{proposition}

\begin{proof}
	We prove only the associativity of the $\uplus$ operator. The rest are analogously proved. Let $\zeta_1, \zeta_2, \zeta_3 \in fPCL(K,P)$ and $\gamma \in fC(P,K)$. Then 
	\begin{align*}
	\left\Vert \zeta_1 \uplus \left(  \zeta_2 \uplus \zeta_3 \right) \right\Vert(\gamma) & = \bigvee_{\gamma = \gamma_1 \cup \gamma^\prime} \left(  \left\Vert \zeta_1 \right\Vert (\gamma_1) \wedge \left\Vert \zeta_2\uplus \zeta_3 \right\Vert(\gamma^\prime) \right) \\ & = \bigvee_{\gamma = \gamma_1 \cup \gamma^\prime} \left(  \left\Vert \zeta_1 \right\Vert (\gamma_1) \wedge \left( \bigvee_{\gamma^\prime =\gamma_2\cup \gamma_3} \left( \left\Vert \zeta_2\right\Vert(\gamma_2) \wedge \left\Vert\zeta_3\right\Vert(\gamma_3) \right)  \right) \right) \\ & =  \bigvee_{\gamma = \gamma_1 \cup \gamma^\prime} \bigvee_{\gamma^\prime =\gamma_2\cup \gamma_3} \left( \left\Vert \zeta_1\right\Vert(\gamma_1) \wedge \left( \left\Vert \zeta_2\right\Vert(\gamma_2)  \wedge \left\Vert \zeta_3\right\Vert(\gamma_3) \right)  \right)  \\ & =  \bigvee_{\gamma = \gamma^\prime \cup \gamma_3} \bigvee_{\gamma^\prime = \gamma_1\cup \gamma_2}  \left(\left( \left\Vert \zeta_1\right\Vert(\gamma_1) \wedge  \left\Vert \zeta_2\right\Vert(\gamma_2) \right) \wedge \left\Vert \zeta_3\right\Vert(\gamma_3)  \right) \\ & =  \bigvee_{\gamma = \gamma^\prime \cup \gamma_3}  \left(\left(\bigvee_{\gamma^\prime = \gamma_1\cup \gamma_2} \left( \left\Vert \zeta_1\right\Vert(\gamma_1) \wedge  \left\Vert \zeta_2\right\Vert(\gamma_2) \right)\right) \wedge \left\Vert \zeta_3\right\Vert(\gamma_3)  \right) \\ & = \bigvee_{\gamma = \gamma^\prime \cup \gamma_3} \left(  \left\Vert \zeta_1 \uplus \zeta_2\right\Vert(\gamma^\prime) \wedge \left\Vert \zeta_3\right\Vert(\gamma_3)  \right) \\ & = \left\Vert \left( \zeta_1\uplus \zeta_2 \right) \uplus \zeta_3 \right\Vert(\gamma)
	\end{align*}
	
	\noindent where the third and fifth equalities hold since $\wedge$ distributes over $\vee$ and the fourth one by the associativity of the $\wedge $ operator.
\end{proof}

\begin{proposition}
	Let $\zeta\in fPCL(K,P)$. Then 
	\[ \left\Vert \sim \zeta \right\Vert(\gamma) = \bigvee_{\gamma^\prime \subseteq \gamma} \left\Vert \zeta \right\Vert(\gamma^\prime) \]
	
	\noindent for every $\gamma\in fC(P,K). $
\end{proposition}

\begin{proof}
	For every $\gamma\in fC(P,K)$ we have
	\begin{align*}
	\left\Vert \zeta \right\Vert (\gamma) & = \bigvee_{\gamma=\gamma^\prime \cup \gamma^{\prime \prime}} \left(  \left\Vert \zeta\right\Vert(\gamma^\prime) \wedge \left\Vert true\right\Vert(\gamma^{\prime \prime}) \right)   \\ & = \bigvee_{\gamma^\prime \subseteq \gamma } \left\Vert \zeta \right\Vert (\gamma^\prime).
	\end{align*}
\end{proof}

\begin{proposition}\label{uplus_over_oplus}
	Let $\zeta, \zeta_1, \zeta_2 \in fPCL(K,P)$. Then
	\[ \zeta \uplus (\zeta_1 \oplus\zeta_2) \equiv (\zeta \uplus \zeta_1) \oplus (\zeta \uplus \zeta_2).   \]
\end{proposition}

\begin{proof}
	For every $\gamma\in fC(P,K)$ we have
	\begin{align*}
	\left\Vert \zeta \uplus (\zeta_1 \oplus \zeta_2) \right\Vert(\gamma) & = \bigvee_{\gamma=\gamma_1\cup \gamma_2} \left( \left\Vert \zeta \right\Vert(\gamma_1) \wedge \left\Vert \zeta_1\oplus \zeta_2\right\Vert(\gamma_2) \right) \\ & =  \bigvee_{\gamma=\gamma_1\cup \gamma_2} \left( \left\Vert \zeta \right\Vert(\gamma_1) \wedge \left(\left\Vert \zeta_1\right\Vert(\gamma_2)\vee \left\Vert \zeta_2\right\Vert(\gamma_2) \right)\right) \\ & =  \bigvee_{\gamma=\gamma_1\cup \gamma_2}\left( \left( \left\Vert \zeta \right\Vert(\gamma_1) \wedge\left\Vert \zeta_1\right\Vert(\gamma_2)\right) \vee \left( \left\Vert \zeta \right\Vert(\gamma_1) \wedge\left\Vert \zeta_2\right\Vert(\gamma_2)\right) \right) \\ & = \bigvee_{\gamma=\gamma_1\cup \gamma_2} \left( \left\Vert \zeta \right\Vert(\gamma_1) \wedge\left\Vert \zeta_1\right\Vert(\gamma_2)\right) \vee \bigvee_{\gamma=\gamma_1\cup \gamma_2} \left( \left\Vert \zeta \right\Vert(\gamma_1) \wedge\left\Vert \zeta_2\right\Vert(\gamma_2)\right)  \\ & = \left\Vert \left(\zeta \uplus \zeta_1\right) \oplus \left(\zeta \uplus \zeta_2 \right)\right\Vert(\gamma)
	\end{align*}
	where the third equality holds since $\wedge $ distributes over $\vee$ and the fourth one by the associativity of $\vee$.
\end{proof}

\begin{proposition}\label{otimes_over_oplus}
	Let $\zeta, \zeta_1, \zeta_2 \in fPCL(K,P)$. Then
	\[ \zeta \otimes (\zeta_1 \oplus\zeta_2) \equiv (\zeta \otimes \zeta_1) \oplus (\zeta \otimes \zeta_2).   \]
\end{proposition}

\begin{proof}
	By the distributivity of $\wedge$ over $\vee$ we get the subsequent proposition.
	
\end{proof}

\begin{proposition}\label{absorpt_fpcl}
	Let $ \zeta \in fPCL(K,P)$. Then
	
	\begin{tabular}{l l l l}
		$(1)$ & $\neg \neg \zeta \equiv \zeta .$ & \hspace*{1.5cm} $(5)$ & $\zeta \oplus false \equiv \zeta.$ \\[0.1cm]  $(2)$ & $\zeta\oplus \zeta \equiv \zeta.$ & \hspace*{1.5cm} $(6)$ & $\zeta \otimes true \equiv \zeta.$ \\[0.1cm] $(3)$ & $\zeta\otimes \zeta \equiv \zeta.$ & \hspace*{1.5cm} $(7)$ & $\zeta \otimes false \equiv false \equiv \zeta \uplus false.$ \\[0.1cm] $(4)$ & $\zeta \oplus true \equiv true.$ & & 
	\end{tabular}

\end{proposition}

\begin{proof} By the properties of De Morgan algebras we can prove the above properties. 
\end{proof}

\begin{proposition}\label{neg}
	Let $ \zeta_1, \zeta_2 \in fPCL(K,P)$. Then
	
	\begin{enumerate}[$(1)$]
		\item $\neg \left( \zeta_1 \oplus \zeta_2 \right) \equiv \left( \neg \zeta_1 \right) \otimes \left( \neg \zeta_2 \right).$
		\item $\neg \left( \zeta_1 \otimes \zeta_2 \right) \equiv \left( \neg \zeta_1 \right) \oplus \left( \neg \zeta_2 \right).$
	\end{enumerate}	
	
\end{proposition}

\begin{proof}
	This proposition holds since $\overline{k\vee k^\prime} = \overline{k} \wedge \overline{k^\prime}$ for every $k,k^\prime \in K$.
\end{proof}

In the following proposition, we present the absorbing laws of our fPCL formulas.

\begin{proposition}\label{absorpt_pcl}
	Let $\zeta, \zeta^\prime \in fPCL(K,P)$. Then 
	\begin{enumerate}[$(1)$]
		\item $\zeta \otimes \left( \zeta \oplus \zeta^\prime \right) \equiv \zeta $.
		
		\item $\zeta \oplus \left( \zeta \otimes \zeta^\prime \right) \equiv \zeta $.
		
	\end{enumerate}
\end{proposition}

\begin{proof}
	This proof is done analogously to the proof of Proposition \ref{absorpt_pil}.
\end{proof}

\begin{proposition}\label{otimes_distib_coale}
	Let $\varphi \in fPIL(K,P)$ and $\zeta_1, \zeta_2 \in fPCL(K,P)$. Then
	\[  \varphi \otimes (\zeta_1 \uplus \zeta_2)  \equiv (\varphi \otimes \zeta_1) \uplus (\varphi \otimes \zeta_2). \] 
\end{proposition}

\begin{proof}
	For every $\gamma\in fC(P,K)$ we compute 
	\begin{align*}
	\left\Vert  \varphi \  \otimes \right. & \left. (\zeta_1 \uplus \zeta_2) \right\Vert (\gamma) \\ & = \left\Vert \varphi \right\Vert(\gamma) \wedge \left\Vert \zeta_1 \uplus \zeta_2 \right\Vert(\gamma) \\ & = \left\Vert \varphi \right\Vert(\gamma) \wedge \bigvee_{\gamma=\gamma_1 \cup \gamma_2} \left(  \left\Vert \zeta_1 \right\Vert(\gamma_1) \wedge \left\Vert \zeta_2\right\Vert(\gamma_2)  \right) \\ & =  \bigvee_{\gamma=\gamma_1 \cup \gamma_2} \left\Vert \varphi \right\Vert(\gamma) \wedge \left(  \left\Vert \zeta_1 \right\Vert(\gamma_1) \wedge \left\Vert \zeta_2\right\Vert(\gamma_2)  \right) \\ & =  \bigvee_{\gamma=\gamma_1 \cup\gamma_2} \left(\bigwedge_{\alpha\in \gamma} \left\Vert \varphi \right\Vert(\alpha) \wedge   \left\Vert \zeta_1 \right\Vert(\gamma_1) \wedge \left\Vert \zeta_2\right\Vert(\gamma_2)  \right) \\ & =  \bigvee_{\gamma=\gamma_1 \cup \gamma_2} \left(\bigwedge_{\alpha_1\in \gamma_1} \left\Vert \varphi \right\Vert(\alpha_1) \wedge \bigwedge_{\alpha_2\in \gamma_2} \left\Vert \varphi \right\Vert(\alpha_2) \wedge  \left\Vert \zeta_1 \right\Vert(\gamma_1) \wedge \left\Vert \zeta_2\right\Vert(\gamma_2)  \right) \\ & =  \bigvee_{\gamma=\gamma_1 \cup \gamma_2} \left( \left\Vert \varphi \right\Vert(\gamma_1) \wedge \left\Vert \varphi \right\Vert(\gamma_2) \wedge  \left\Vert \zeta_1 \right\Vert(\gamma_1) \wedge \left\Vert \zeta_2\right\Vert(\gamma_2)  \right) \\ & =  \bigvee_{\gamma=\gamma_1 \cup \gamma_2} \left(  \left\Vert \varphi \otimes \zeta_1 \right\Vert(\gamma_1) \wedge  \left\Vert \varphi \otimes \zeta_2 \right\Vert(\gamma_2) \right) \\ & = \left\Vert \left( \varphi \otimes \zeta_1\right) \uplus \left( \varphi \otimes \zeta_2  \right) \right\Vert(\gamma)
	\end{align*}
	where the third equality holds since $\wedge $ distributes over $\vee$ and the fifth one by the idempotency and associativity of the $\wedge$ operator. \end{proof}

\begin{proposition}\label{pil_prop}
	Let $\varphi \in fPIL(K,P)$. Then
	\begin{flushleft}
		$\begin{array}{l l l}
		(1) \  \varphi \uplus \varphi \equiv \varphi. & \hspace*{1.2cm} (2) \ \neg \sim \varphi \equiv \  ! \varphi.  & \hspace*{1.2cm} (3) \ \neg \varphi \equiv \ \sim \  ! \varphi.
		\end{array}$
	\end{flushleft}
	
\end{proposition}
\begin{proof}
	For every $\gamma\in fC(P,K)$ we have 
	\begin{enumerate}[(1)]
		\item \begin{align*}
		\left\Vert \varphi \uplus \varphi \right\Vert(\gamma) & = \bigvee_{\gamma=\gamma_1\cup \gamma_2} \left(\left\Vert\varphi \right\Vert(\gamma_1) \wedge \left\Vert\varphi \right\Vert(\gamma_2)\right) \\ & = \bigvee_{\gamma=\gamma_1 \cup\gamma_2} \left( \bigwedge_{\alpha_1\in \gamma_1} \left\Vert \varphi\right\Vert(\alpha_1) \wedge \bigwedge_{\alpha_2\in \gamma_2}\left\Vert \varphi\right\Vert(\alpha_2) \right) \\ & = \bigvee_{\gamma=\gamma_1 \cup\gamma_2} \left( \bigwedge_{\alpha\in \gamma_1\cup \gamma_2} \left\Vert \varphi\right\Vert(\alpha) \right)\\ & = \bigwedge_{\alpha\in \gamma} \left\Vert\varphi\right\Vert(\alpha)   = \left\Vert \varphi \right\Vert(\gamma)
		\end{align*}
		where the third and fourth equalities hold since the operators $\wedge$ and $ \vee $ are idempotent.
		
		\item 
		\begin{align*}
		\left\Vert \neg \sim \varphi \right\Vert (\gamma) & =   \left\Vert \neg \left( \varphi \uplus true \right) \right\Vert (\gamma) \\ & = \overline{\left\Vert  \varphi \uplus true  \right\Vert (\gamma) } \\ &  = \overline{\bigvee_{\gamma=\gamma_1\cup \gamma_2}\left(\left\Vert  \varphi \right\Vert(\gamma_1) \wedge \left\Vert  true  \right\Vert (\gamma_2) \right) } \\ &  = \overline{\bigvee_{\gamma^\prime \subseteq \gamma}\left\Vert  \varphi \right\Vert(\gamma^\prime) } \\ &  = \overline{\bigvee_{\gamma^\prime \subseteq \gamma} \bigwedge_{\alpha\in \gamma^\prime}\left\Vert  \varphi \right\Vert(\alpha) } \\ &  = \bigwedge_{\gamma^\prime \subseteq \gamma} \bigvee_{\alpha\in \gamma^\prime}\overline{\left\Vert  \varphi \right\Vert(\alpha) }.
		\end{align*}
		
		\noindent Let $\gamma=\{ \alpha_1, \dots, \alpha_n \}$ and $\{\alpha_1\},\dots, \{\alpha_n\}, \gamma_1, \dots, \gamma_k $ be all possible subsets of $\gamma$, where $k=2^n - (n+1)$ and $|\gamma_i|>1$ for every $i\in \{1, \dots,k\}$. Therefore, we get
		\begin{align*}
		\left\Vert \neg(\varphi \uplus 1) \right\Vert(\gamma) & = \bigwedge_{\gamma^\prime \subseteq \gamma} \bigvee_{\alpha\in \gamma^\prime}\overline{\left\Vert  \varphi \right\Vert(\alpha) } \\ & = \left(\bigwedge_{i=1}^n \bigvee_{\alpha\in \{a_i\}} \overline{\left\Vert \varphi \right\Vert(\alpha)} \right)\wedge \left(  \bigwedge_{j=1}^k \bigvee_{\alpha\in \gamma_j} \overline{\left\Vert\varphi \right\Vert(\alpha)} \right) \\ & = \left(\overline{\left\Vert \varphi\right\Vert(\alpha_1)} \wedge \dots \wedge  \overline{\left\Vert \varphi\right\Vert(\alpha_n)} \right)\wedge \left(  \bigwedge_{j=1}^k \bigvee_{\alpha\in \gamma_j} \overline{\left\Vert\varphi \right\Vert(\alpha)} \right) \\ & = \left(\bigwedge_{\alpha\in \gamma}\overline{\left\Vert \varphi\right\Vert(\alpha)} \right)\wedge \left(  \bigwedge_{j=1}^k \bigvee_{\alpha^\prime\in \gamma_j} \overline{\left\Vert\varphi \right\Vert(\alpha^\prime)} \right) \\ & = \left(\bigwedge_{\alpha\in\gamma}\left\Vert !  \varphi\right\Vert(\alpha) \right)\wedge \left(  \bigwedge_{j=1}^k \bigvee_{\alpha^\prime\in \gamma_j} \left\Vert !\varphi \right\Vert(\alpha^\prime) \right) \\ & = \bigwedge_{j=1}^k \bigvee_{\alpha^\prime\in \gamma_j} \left(  \bigwedge_{\alpha\in \gamma}\left\Vert ! \varphi\right\Vert(\alpha) \wedge \left\Vert ! \varphi \right\Vert(\alpha^\prime) \right) \\ & = \bigwedge_{\alpha\in \gamma}\left\Vert ! \varphi\right\Vert(\alpha) \\ & = \left\Vert ! \varphi\right\Vert(\gamma)
		\end{align*}
		
		\noindent where for the validity of the last equality we give the following explanation. For every $j\in \{1, \dots, k\}$ and for every $\alpha^\prime\in \gamma_j$ we have $\bigwedge_{\alpha\in \gamma}\left\Vert ! \varphi\right\Vert(\alpha) \wedge \left\Vert ! \varphi \right\Vert(\alpha^\prime) =  \bigwedge_{\alpha\in \gamma}\left\Vert ! \varphi\right\Vert(\alpha)$ since $\alpha^\prime \in \gamma$ and $\wedge $ is idempotent. Lastly, by the idempotency of $\vee$ we get the last equality.

		\item 	\begin{align*}
		\left\Vert \sim \ !  \varphi  \right\Vert (\gamma) & = \left\Vert \neg \left( \neg \sim \ ! \varphi  \right) \right\Vert (\gamma) \\ & = \left\Vert \neg ! (! \varphi) \right\Vert (\gamma) \\ & = \left\Vert \neg \left( ! ! \varphi \right) \right\Vert(\gamma) \\ & = \left\Vert \neg \varphi \right\Vert(\gamma)
		\end{align*}
		where the second equality holds by Proposition \ref{pil_prop}(2).
		
	\end{enumerate}
\end{proof}

\begin{proposition}\label{fpil_ovee_to_fpcl}
	Let $J$ be a finite index set and $\varphi_j $ a \emph{fPIL} formula over $P$ and $K$ for every $j\in J$. Let $\gamma\in fC(P,K)$. Then  
	\[  \left\Vert  \underset{j\in J}{\bigovee} \varphi_j  \right\Vert(\gamma)  = \bigvee_{J^\prime \subseteq J} \bigvee_{ \underset{j^\prime \in J^\prime}{\bigcupdot}\gamma_{j^\prime} =\gamma  }  \left( \bigwedge_{j^\prime \in J^\prime} \left\Vert \varphi_{j^\prime} \right\Vert(\gamma_{j^\prime})   \right)   \]
	
	\noindent where $\cupdot$ denotes the disjoint union of sets.
	
\end{proposition}

\begin{proof}
	Let $\gamma = \{ \alpha_1, \dots, \alpha_n \} \in fC(P,K)$ and a finite index set $J$. Then 
	\begin{align*}
	\left\Vert  \underset{j\in J}{\bigovee} \varphi_j  \right\Vert(\gamma) & = \bigwedge_{\alpha\in \gamma} \left\Vert \underset{j\in J}{\bigovee} \varphi_j \right\Vert(\alpha) \\ & = \bigwedge_{\alpha\in \gamma} \left( \bigvee_{j\in J} \left\Vert \varphi_j \right\Vert(\alpha) \right) \\ & = \bigvee_{(j_1, \dots, j_n) \in J^n} \left(  \left\Vert \varphi_{j_1} \right\Vert(\alpha_1) \wedge \left\Vert \varphi_{j_2} \right\Vert(\alpha_2) \wedge \dots \wedge \left\Vert \varphi_{j_n} \right\Vert(\alpha_n) \right)  \\ & =  \bigvee_{J^\prime \subseteq J} \bigvee_{ \underset{j^\prime \in J^\prime}{\bigcupdot} \gamma_{j^\prime} = \gamma} \left( \bigwedge_{j^\prime\in J^\prime} \left\Vert \varphi_{j^\prime} \right\Vert(\gamma_{j^\prime}) \right)
	\end{align*}
	
	\noindent  where for the validity of the last equality we give the following explanation. In the third equality we have all possible n-tuples with elements from $J$. Hence, there are cases where we have repetitions of some $j_{i}$'s. Moreover, in every parenthesis in the third equality, each $\alpha_{j}$ appears exactly once and therefore we get the disjoint unions of sets that are equal to $\gamma$ in the last equality. Lastly, considering the above, the commutativity of the operators $\wedge, $ $\vee$ and the idempotency of the $\vee$ operator, we get the last equality.  
	
\end{proof}

\begin{proposition}\label{pil_coal_n}
	Let $J$ be a finite index set and $\varphi_j $ a \emph{fPIL} formula over $P$ and $K$ for every $j\in J$. Then 
	\[ \biguplus_{j\in J} \varphi_j \equiv \bigotimes_{j\in J} \left( \sim \varphi_j \right) \otimes \left( \underset{j\in J}{\bigovee}  \ \varphi_j  \right).  \]
\end{proposition}

\begin{proof}
	Let $\gamma\in fC(P,K)$ and $J=\{ 1, \dots, n \}$ a finite index set. Then we get 
	\begin{align*}
	& \left\Vert \bigotimes_{j\in J} \left( \sim \varphi_j \right) \otimes \left( \underset{j\in J}{\bigovee} \ \varphi_j  \right) \right\Vert(\gamma) \\& = \bigwedge_{j\in J}\left( \bigvee_{\gamma_j\subseteq \gamma} \left\Vert \varphi_j\right\Vert(\gamma_j) \right)\wedge \left( \bigvee_{J^\prime \subseteq J} \bigvee_{ \underset{j^\prime \in J^\prime}{\bigcupdot}\gamma_{j^\prime}^\prime = \gamma   }  \left( \bigwedge_{j^\prime \in J^\prime} \left\Vert \varphi_{j^\prime} \right\Vert(\gamma_{j^\prime}^\prime)   \right) \right)  \\ & = \bigvee_{\gamma_1\subseteq \gamma} \ldots \bigvee_{\gamma_n\subseteq \gamma}\left(  \bigwedge_{j\in J}\left\Vert \varphi_j\right\Vert(\gamma_j)    \right) \wedge \left( \bigvee_{J^\prime \subseteq J} \bigvee_{ \underset{j^\prime \in J^\prime}{\bigcupdot}\gamma_{j^\prime}^\prime =\gamma  }  \left( \bigwedge_{j^\prime \in J^\prime} \left\Vert \varphi_{j^\prime} \right\Vert(\gamma_{j^\prime}^\prime)   \right) \right) \\ & = \bigvee_{\gamma_1\subseteq \gamma} \ldots \bigvee_{\gamma_n\subseteq \gamma}\bigvee_{J^\prime \subseteq J} \bigvee_{ \underset{j^\prime \in J^\prime}{\bigcupdot}\gamma_{j^\prime}^\prime =\gamma   } \left( \left\Vert \varphi_1\right\Vert(\gamma_1) \wedge \dots \wedge  \left\Vert \varphi_n\right\Vert(\gamma_n) \wedge  \bigwedge_{j^\prime \in J^\prime} \left\Vert \varphi_{j^\prime} \right\Vert(\gamma_{j^\prime}^\prime)\right) \\ & = \bigvee_{\gamma_1\subseteq \gamma} \ldots \bigvee_{\gamma_n\subseteq \gamma}\bigvee_{J^\prime \subseteq J} \bigvee_{ \underset{j^\prime \in J^\prime}{\bigcupdot}\gamma_{j^\prime}^\prime =\gamma   } \left( \bigwedge_{j^\prime \in J^\prime} \left\Vert \varphi_{j^\prime } \right\Vert(\gamma_{j^\prime}^\prime \cup \gamma_{j^\prime}) \wedge \bigwedge_{j\in J\backslash J^\prime} \left\Vert \varphi_j \right\Vert(\gamma_j) \right)  
	\end{align*}
	
	\noindent where the first equality holds by Proposition \ref{fpil_ovee_to_fpcl} and the fourth one by the idempotency of $\wedge$. We observe that the sets $\gamma_{j^\prime}^\prime \cup \gamma_{j^\prime}$ and $\gamma_{j} $ for every $j^\prime \in J^\prime$ and $j\in J\backslash J^\prime$, consist all possible subsets of $\gamma$ where the union of them is equal to $\gamma$. So, by the idempotency of $\wedge$ and $\vee$ we get
	\begin{align*}
	\left\Vert \bigotimes_{j\in J} \left( \sim \varphi_j \right) \otimes \left( \underset{j\in J}{\bigovee} \ \varphi_j  \right) \right\Vert(\gamma) & = \bigvee_{\gamma_1^{\prime \prime } \cup \dots \cup \gamma_n^{\prime \prime}=\gamma} \left( \left\Vert \varphi_1\right\Vert(\gamma_1^{\prime \prime}) \wedge \dots \wedge  \left\Vert \varphi_n\right\Vert(\gamma_n^{\prime \prime})\right) \\ & = \left\Vert \varphi_1\uplus \dots \uplus  \varphi_n \right\Vert (\gamma).
	\end{align*} \end{proof}

\begin{proposition}\label{pil_coal_neg}
	Let $J$ be a finite index set and $\varphi_j$ a \emph{fPIL} formula for every $j\in J$. Then 
	\[  \neg \biguplus_{j\in J} \varphi_j \equiv \bigoplus_{j\in J} \left(  ! \varphi_j \right) \oplus \sim \left( \underset{j\in J}{\bigowedge}  ! \varphi_j  \right)    \] 
\end{proposition}

\begin{proof} We get  \begin{align*}
	\neg \biguplus_{j\in J} \varphi_j & \equiv \neg \left( \bigotimes_{j\in J} \left( \sim \varphi_j \right) \otimes \left( \underset{j\in J}{\bigovee}  \varphi_j  \right)   \right) \\ & \equiv \bigoplus_{j\in J} \left(  \neg \sim \varphi_j \right) \oplus  \neg \left( \underset{j\in J}{\bigovee} \ \varphi_j  \right)    \\ & \equiv \bigoplus_{j\in J} \left( ! \varphi_j \right) \oplus \sim \ ! \left( \underset{j\in J}{\bigovee} \varphi_j  \right)  \\ & \equiv \bigoplus_{j\in J} \left( ! \varphi_j \right) \oplus \sim \underset{j\in J}{\bigowedge} \left(! \varphi_j\right) 
	\end{align*}
	\noindent where the third equivalence holds by Proposition \ref{pil_prop}. 
\end{proof}

\begin{proposition}\label{otimes_to_uplus}
	Let $J$ be a finite index set and $\varphi_j$ a \emph{fPIL} formula for every $j\in J$. Then 
	\[ \bigotimes_{j\in J} (\sim \varphi_j)  \equiv \  \sim  \biguplus_{j\in J} \varphi_j . \]
\end{proposition}

\begin{proof}
	Let $\gamma\in fC(P,K)$. Then 
	\begin{align*}
	\left\Vert \bigotimes_{j\in J} (\sim \varphi_j)  \right\Vert(\gamma) & = \bigwedge_{j\in J} \left(  \bigvee_{\gamma_j\subseteq \gamma} \left\Vert \varphi_j\right\Vert(\gamma_j)  \right)  \\ & = \bigvee_{\bigcup_{j\in J}\gamma_j \subseteq \gamma} \left( \bigwedge_{j\in J} \left\Vert\varphi_j\right\Vert(\gamma_j) \right) \\ & = \bigvee_{\gamma^\prime \subseteq \gamma} \left(\bigvee_{\bigcup_{j\in J}\gamma_j = \gamma^\prime} \left( \bigwedge_{j\in J} \left\Vert\varphi_j\right\Vert(\gamma_j) \right) \right) \\ & = \bigvee_{\gamma^\prime \subseteq \gamma} \left\Vert \biguplus_{j\in J} \varphi_j \right\Vert(\gamma^\prime) \\ & = \left\Vert \sim \biguplus_{j\in J} \varphi_j \right\Vert(\gamma)
	\end{align*}
	\noindent where the second equality holds since $\wedge$ distributes over $\vee$.
\end{proof}

\begin{proposition}\label{pil_coal_conj}
	Let $J$ and $K$ be finite index sets and $\varphi_j$, $\varphi_k^\prime$ \emph{fPIL} formulas for every $j\in J$ and $k\in K$. Then 
	\[ \left(\biguplus_{j\in J} \varphi_{j}\right) \otimes \left( \biguplus_{k\in K} \varphi_{k}^\prime  \right) \equiv \ \sim \left( \biguplus_{j\in J} \varphi_j \uplus \biguplus_{k\in K} \varphi_k^\prime \right) \otimes \left( \underset{(j,k)\in J\times K}{\bigovee} \left( \varphi_j \sowedge \varphi_k^\prime \right)   \right)  \]
\end{proposition}

\begin{proof}
	By Proposition \ref{absorpt_fpcl}(1) we get
	\begin{align*}
	\left(\biguplus_{j\in J} \varphi_{j}\right) \otimes \left( \biguplus_{k\in K} \varphi_{k}^\prime \right) \equiv \neg \neg \left(\left(\biguplus_{j\in J} \varphi_{j}\right) \otimes \left( \biguplus_{k\in K} \varphi_{k}^\prime \right) \right). 
	\end{align*}
	
	\noindent By Proposition \ref{pil_coal_neg} it is valid that
	\begin{align*}
	\neg \left( \left(\biguplus_{j\in J} \varphi_{j}\right) \otimes \right. & \left. \left( \biguplus_{k\in K} \varphi_{k}^\prime  \right) \right) \\ & \equiv \neg  \left( \biguplus_{j\in J} \varphi_{j} \right) \oplus \neg \left( \biguplus_{k\in K} \varphi_{k}^\prime \right) \\ & \equiv \bigoplus_{j\in J} \left(  ! \varphi_j \right) \oplus \sim \left( \underset{j\in J}{\bigowedge}  ! \varphi_j  \right) \oplus \bigoplus_{k \in K} \left(  ! \varphi_k^\prime  \right) \oplus \sim \left( \underset{k\in K}{\bigowedge}  ! \varphi_k^\prime  \right)
	\end{align*}
	
	\noindent and so
	\begin{align*}
	\neg \neg \left( \left(\biguplus_{j\in J} \right. \right.  & \left. \left.\varphi_{j\in J}\right) \otimes \left( \biguplus_{k\in K} \varphi_{k}^\prime  \right) \right) \\ & \equiv \neg \left( \bigoplus_{j\in J} \left(  ! \varphi_j \right) \oplus \sim \left( \bigowedge_{j\in J}  ! \varphi_j  \right) \oplus \bigoplus_{k \in K} \left(  ! \varphi_k^\prime  \right) \oplus \sim \left( \underset{k\in K}{\bigowedge}   ! \varphi_k^\prime  \right)  \right) \\ & \equiv \bigotimes_{j\in J} \left( \neg ! \varphi_j \right) \otimes \neg\sim \left( \bigowedge_{j\in J}  ! \varphi_j  \right) \otimes \bigotimes_{k\in K} \left( \neg ! \varphi_{k}^\prime  \right) \otimes \neg \sim \left( \underset{k\in K}{\bigowedge}   ! \varphi_k^\prime  \right)  \\ & \equiv \bigotimes_{j\in J} \left( \sim  \varphi_j \right) \otimes ! \left( \underset{j\in J}{\bigowedge}  !\varphi_j  \right) \otimes \bigotimes_{k\in K} \left(\sim  \varphi_{k}^\prime  \right) \otimes ! \left( \underset{k\in K}{\bigowedge} ! \varphi_k^\prime  \right) \\ & \equiv \bigotimes_{j\in J} \left( \sim  \varphi_j \right) \otimes  \left(\underset{j\in J}{\bigovee} \varphi_j  \right) \otimes \bigotimes_{k\in K} \left(\sim  \varphi_{k}^\prime  \right) \otimes  \left( \underset{k\in K}{\bigovee}   \varphi_k^\prime  \right)  \\ &  \equiv  \bigotimes_{j\in J} \left( \sim  \varphi_j \right) \otimes  \bigotimes_{k\in K} \left(\sim  \varphi_{k}^\prime  \right) \otimes \left(\underset{j\in J}{\bigovee}  \varphi_j  \right) \otimes \left( \underset{k\in K}{\bigovee}   \varphi_k^\prime  \right) \\ &  \equiv  \bigotimes_{j\in J} \left( \sim  \varphi_j \right) \otimes  \bigotimes_{k\in K} \left(\sim  \varphi_{k}^\prime  \right) \otimes \left(\left(\underset{j\in J}{\bigovee}  \varphi_j  \right) \sowedge \left( \underset{k\in K}{\bigovee}   \varphi_k^\prime  \right) \right)   \\ &  \equiv \sim \left( \biguplus_{j\in J} \varphi_j \uplus \biguplus_{k\in K} \varphi_k^\prime \right) \otimes \left(\underset{(j,k)\in J\times K}{\bigovee}   \left( \varphi_j \sowedge \varphi_k^\prime \right)   \right)
	\end{align*}
	
	\noindent where the third equivalence holds by Proposition \ref{pil_prop} and the last one by Proposition \ref{otimes_to_uplus}.
\end{proof}

We proceed with an important property of fPCL formulas over $P$ and a Kleene algebra.

\begin{proposition}\label{kleene_ports}
	Let $P$ be a set of ports and $K_\mathbf{3}$ a Kleene algebra. Then 
	\[ \left(  p \ \otimes \ ! p \right) \otimes \left(q  \sovee  ! q \right) \equiv_{K_\mathbf{3}}  p \ \otimes \ ! p   \]
	
	\noindent where $p,q \in P.$
	
\end{proposition}

\begin{proof}
	Let $\gamma\in fC(P,K_\mathbf{3})$. Then 
	\begin{align*}
	\left\Vert \left(  p\otimes ! p \right) \otimes \left(q  \sovee ! q \right) \right\Vert(\gamma) & = \bigwedge_{\alpha\in \gamma} \left\Vert \left(  p\otimes ! p \right) \otimes \left(q  \sovee ! q \right)\right\Vert(\alpha) \\ & = \bigwedge_{\alpha\in \gamma} \left( \left(\alpha(p)\wedge \overline{\alpha(p)} \right) \wedge \left( \alpha(q) \vee \overline{\alpha(q)} \right)\right) \\ & = \bigwedge_{\alpha\in \gamma} \left(\alpha(p)\wedge \overline{\alpha(p)} \right)  \\ & = \left\Vert p \otimes ! p\right\Vert(\gamma)
	\end{align*}
	where the third equality holds since $\left( k\wedge \overline{k} \right) \wedge \left( k^\prime \vee \overline{k^\prime} \right) = k\wedge \overline{k}$ for every $k,k^\prime\in K_{\mathbf{3}}$.
\end{proof}

Next, we give properties of $fPCL$ formulas over $P$ and a Boolean algebra.

\begin{proposition}\label{boolean_prop}
	Let $\varphi$ and $\zeta$ be a \emph{fPIL} and a \emph{fPCL} formula, respectively, over $P$ and a Boolean algebra. Then

	\begin{tabular}{l l l l}
		$(1)$ & $\varphi \ \otimes \ ! \varphi \equiv_{\mathbf{B}} false .$ & \hspace*{1.5cm} $(3)$ &  $\zeta \otimes \neg \zeta \equiv_{\mathbf{B}}  false.$ \\[0.1cm]  $(2)$ & $\varphi   \sovee  ! \varphi \equiv_{\mathbf{B}} true .$ & \hspace*{1.5cm} $(4)$ & $ \zeta \oplus \neg \zeta \equiv_{\mathbf{B}}  true. $ 
	\end{tabular}

\end{proposition}

\begin{proof}
	By the properties of the Boolean algebra we get the proposition. 
\end{proof}

Let $P$ be a finite set of ports. By the previous results we get that $p \ \otimes \ !p \equiv_{\mathbf{2}} false $ but $p \ \otimes \ !p \not \equiv_{\mathbf{3}} false$ for every $p\in P.$

\section{Examples}

\begin{figure}[t]
	\begin{center}
		\begin{tikzpicture}[scale=0.6]
		
		\filldraw[color=black!60, fill=black!5, very thick](0,0) circle (0.8);
		\node at (0,0.28) { $C_1$};
		
		\filldraw[color=black!60, fill=white!5, very thick](-0.4,-0.3) circle (0.3);
		\node at (-0.4,-0.3) { $r_1$};
		\draw[fill] (-0.6,-0.5) circle [radius=2pt];
		
		\filldraw[color=black!60, fill=white!5, very thick](0.4,-0.3) circle (0.3);
		\node  at (0.4,-0.3) { $s_1$};
		\draw[fill]  (0.6,-0.5) circle [radius=2pt];
		
		\filldraw[color=black!60, fill=black!5, very thick](-2.5,-2.5) circle (0.8);
		\node at (-2.8,-2.5) { $C_2$};
		
		\filldraw[color=black!60, fill=white!5, very thick](-2.1,-2.8) circle (0.3);
		\node at (-2.1,-2.8) { $r_2$};
		\draw[fill] (-1.9,-3) circle [radius=2pt];
		
		\filldraw[color=black!60, fill=white!5, very thick](-2.25,-2.05) circle (0.3);
		\node at (-2.25,-2.05) { $s_2$};
		\draw[fill] (-2.1,-1.8) circle [radius=2pt];

		\filldraw[color=black!60, fill=black!5, very thick](2.5,-2.5) circle (0.8);
		\node at (2.8,-2.5) { $C_3$};
		
		\filldraw[color=black!60, fill=white!5, very thick](2.1,-2.8) circle (0.3);
		\node at (2.1,-2.8) { $s_3$};
		\draw[fill] (1.9,-3) circle [radius=2pt];
		
		\filldraw[color=black!60, fill=white!5, very thick](2.25,-2.05) circle (0.3);
		\node at (2.25,-2.05) { $r_3$};
		\draw[fill] (2.1,-1.8) circle [radius=2pt];
		
		\filldraw[color=black!60, fill=black!5, very thick](0,-4) circle (0.8);
		\node at (0,-4.3) { $C_4$};
		
		\filldraw[color=black!60, fill=white!5, very thick](-0.4,-3.7) circle (0.3);
		\node at (-0.4,-3.7) { $s_4$};
		\draw[fill] (-0.6,-3.5) circle [radius=2pt];
		
		\filldraw[color=black!60, fill=white!5, very thick](0.4,-3.7) circle (0.3);
		\node  at (0.4,-3.7) { $r_4$};
		\draw[fill] (0.6,-3.5) circle [radius=2pt];
		
		
		\draw (0.6,-0.5)--(0.6,-3.5);
		
		\draw (0.6,-0.5)--(2.1,-1.8);
		
		\draw (-2.1,-1.8) --(2.1,-1.8);
		\draw (-2.1,-1.8) --(-0.6,-0.5);

		\end{tikzpicture}
	\end{center}
	\caption{Peer-to-peer architecture.}
	\label{p2p:image}
\end{figure}
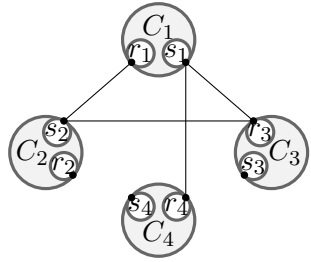

\begin{example}[Peer-to-Peer architecture]

	Peer-to-peer architecture (P2P for short) is a commonly used computer networking architecture. All peers in the architecture have some available resources, such as processing power and network bandwidth (cf. \cite{p2p_def}). Those resources are available to the other peers that participate in the architecture without the need of a central component to coordinate their interactions. This is not the case in Request/Response architecture where a coordinator is needed (cf. \cite{Ma:Co}). All peers in the architecture are both suppliers and consumers of their resources and so there is no distinction between them.

	In our example, we consider four components $C_1, C_2, C_3$ and $C_4$ (Figure \ref{p2p:image}). Every component has two ports denoted by $r$ and $s$ which represent, respectively, the functions receive and send. Let $J=\{1,2,3,4\}$. So, the set of ports is $P= \bigcup_{j\in J}\{r_j,s_j\}$. Each component can receive and send information to as many other components in the architecture except from itself. One possible architecture scheme is shown in Figure \ref{p2p:image}. In the sequel we construct a fPCL formula which describes the P2P architecture with four components with a fPCL formula.

	Firstly, let two distinct components $C_j$ and $C_{j^\prime}$ where $j,j^\prime \in J$. The interaction between $C_j$ and $C_{j^\prime}$ where $C_j$ rceives information from $C_{j^\prime}$, is characterised by the following \emph{fPIL} formula
	\[  \varphi_{j,j^\prime} =  r_j\otimes s_{j^\prime } \otimes ! s_j \otimes ! r_{j^\prime} \otimes \bigotimes_{j^{\prime \prime} \in J \backslash\{ j,j^\prime \} } \left( ! r_{j^{\prime \prime}} \otimes ! s_{j^{\prime \prime}} \right)  \]
	\noindent Next, as it was mentioned above, $C_j$ can receive information from more than one components. Let $J^\prime \subseteq J\backslash \{j \} $. The interactions between $C_{j}$ and $C_{j^\prime }$, where $j^\prime \in J^\prime$ are characterized by $  \zeta_{j, J^\prime} = \biguplus_{j^\prime \in J^\prime} \varphi_{j,j^\prime}.$ However, $J^\prime $ can be any non empty subset of $J\backslash\{ j \}$. So the \emph{fPCL} formula 
	\[  \zeta_j  = \bigoplus_{J^\prime \in \mathcal{P}(J\backslash\{j\})\backslash \{ \emptyset\} }  \zeta_{j, J^\prime} \]
	\noindent, where $\mathcal{P}(J\backslash\{j\})$ denotes the power set of $J\backslash\{j\}$, describes all possible architecture schemes between $C_j$ and the rest components in the architecture. Lastly, some components may not interact at all with the others. Therefore, the \emph{fPCL} formula
	\[ \zeta =  \bigoplus_{J^{\prime\prime} \in \mathcal{P}(J)\backslash\{ \emptyset \}} \biguplus_{j \in J^{\prime\prime}} \zeta_{j} \]
	\noindent describes all possible architecture schemes of the P2P architecture with four components.

	In our example, we consider that every port has a degree of uncertainty. Let the fuzzy algebra and a configuration set $\gamma\in fC(P,\mathbf{F})$. For every $\alpha\in fI(P,\textbf{F})$ the value $\alpha(p)$ represents the degree of uncertainty of the port $p\in P.$ If $\alpha(p) = 0$ then the port has an absolute uncertain bahavior. If $\alpha(p) = 1 $ the port will participate with no uncertainty, i.e., it will participate with no fault in its behavior. Then the value $\left\Vert \sim \zeta \right\Vert(\gamma)$ gives the maximum uncertainty that can occur in the architecture considering the given interactions of $\gamma$.

\end{example}

\begin{example}
	We recall from \cite{Ma:Co} the Master/Slave architecture for two masters $M_1, M_2$ and two slaves $S_1, S_2$ with ports $m_1, m_2$ and $s_1, s_2$, respectively. Masters can interact only with slaves, and vice versa. Each slave can interact with only one master.
	
	As it was mentioned in the Introduction, software architectures have a degree of uncertainty. We show how we can compute the uncertainty of the Master/Slave architecture over a finite number of components. We consider the fuzzy algebra and the set of ports $P=\{s_1, m_1, s_2, m_2 \}$.  
	
	Next, we construct the \emph{fPCL} formula which describes the architecture. The interaction between a master $m \in \{m_1, m_2\}$ and a slave $s\in \{s_1, s_2\}$ is described by the \emph{fPIL} formula 
	\[ \varphi_{s,m} = s\otimes m\otimes \ ! s^\prime \otimes \ ! m^\prime \]
	\noindent  where $s^\prime \not = s$ and $m^\prime \not = m$. Moreover, as it was mentioned above, every master can interact with only one slave, and vice versa. Hence, the \emph{fPCL} formula 
	\[ \zeta = \left(  \varphi_{s_1,m_1} \oplus \varphi_{s_2, m_1} \right) \uplus \left(  \varphi_{s_1,m_2} \oplus \varphi_{s_2, m_2} \right) \]
	\noindent describes the Master /Slave architecture with two masters and two slaves. Let $\gamma\in fC(P,K_\mathbf{F})$ be a set  of estimations of uncertainty of the ports in the architecture and the \emph{fPCL} formula $\sim \zeta$. The value $\left\Vert \sim \zeta \right\Vert (\gamma)$ gives the maximum value among the values that represent the maximum uncertainty among the architecture patterns.

\end{example}

\section{Normal Form and Decidability of Equivalence} \label{section_normal_form}

In this section we examine the decidability of equivalence of fPCL formulas. Let $\zeta $ and $\zeta^\prime$ be fPCL formulas over the set of ports $P=\{p,q,r\}$ and the fuzzy algebra $\mathbf{F}$. By Definition \ref{fpcl_equiv}, $\zeta \equiv_{\textbf{F}} \zeta^\prime$ if $\left\Vert \zeta \right\Vert(\gamma) = \left\Vert\zeta^\prime \right\Vert(\gamma)$ for every $\gamma\in fC(P,\mathbf{F})$. However, the set $fC(P,\mathbf{F})$ is infinite and so it is impossible to check the equivalence by the previous way. This is not the case for fPCL formulas over De Morgan algebras with finite set $K$ such as the two element Boolean algebra and the three element Kleene algebra. However, if we prove that two fPCL formulas have the same normal form, then they are equivalent.

In the sequel, we show that every fPCL formula over $P$ and a Kleene algebra $K_\mathbf{3}$, can be equivalently written in a normal form. Consequently, we show that the equivalence problem for fPCL formulas over $P$ and a Kleene algebra $K_\mathbf{3}$ is decidable. In the following, we give some useful definitions for the definition of the normal form of our fPCL formulas.  

\begin{definition}
	Let $P$ be a set of ports. A \emph{fPIL} formula $\varphi $ is called f-monomial if it is of the form 
	$$ \varphi = \underset{p_1\in P_1}{\bigowedge} p_1 \sowedge \underset{p_2\in P_2}{\bigowedge} ! p_2.$$
	\noindent where $P_1, P_2 \subseteq P $ and $P_1\cup P_2 \not = \emptyset$.
\end{definition}

Following the previous definition, $P_1 \cap P_2 $ can be either empty or not. Consider $P=\{ p,q,r \}$ be a set of ports. The fPIL formulas $p  \sowedge   ! p  \sowedge  ! q$ and $p\sowedge r$ are f-monomials.

\begin{definition}
	Let $P$ be a set of ports and $K$ a De Morgan algebra. A \emph{fPIL} formula $\varphi$ is said to be in $fpil$-normal form if it is of the form
	\begin{enumerate}[$(1)$]
		\item $\varphi \ \dot{\equiv} \ \underset{i\in I}{\bigovee}  \ \varphi_i $, where $I$ is a finite index set, $\varphi_i$ is a f-monomial for every $i\in I$ and $\varphi_i \dot{\not \equiv} \varphi_{i^\prime}$ for every $i, i^\prime \in I$ with $i\not = i^\prime$, or
		
		\item $\varphi \ \dot{\equiv} \ true$, or
		\item $\varphi \ \dot{\equiv} \ false.$
	\end{enumerate}

\end{definition}

\begin{definition}
	Let $P$ be a finite set of ports and $K$ a De Morgan algebra. A \emph{fPCL} formula $\zeta$ over $P$ and $K$ is said to be in normal form if it is of the following form:
	\begin{enumerate}[$(1)$]
		\item $\zeta = \bigoplus_{i\in I}\biguplus_{j\in J_i} \varphi_{i,j} $, where $I, J_i$ are finite index sets for every $i\in I$ and $\varphi_{i,j}\not \equiv false$ is in $fpil$-normal form for every $i\in I$ and $j\in J_i$, or
		\item $\zeta = true$, or
		\item $\zeta=false.$
	\end{enumerate}
\end{definition}

By Propositions \ref{absorpt_pil}, \ref{absorpt_fpcl}, \ref{pil_prop},  for every fPCL formula in normal form we can construct its equivalent one in normal form satisfying the following statements: 
\begin{enumerate}[(1)]
	\item  Let $i\in I$. Then $\varphi_{i, j} \not \equiv \varphi_{i, j^\prime}$ for every $j \not = j^\prime$.
	
	\item Let $i,i^\prime \in I$ with $i\not = i^\prime$. Then $\biguplus_{j\in J_i} \varphi_{i,j} \not \equiv \biguplus_{j\in J_{i^\prime}} \varphi_{i^\prime,j}$.
\end{enumerate}

\begin{quotation}
	In the sequel, every fPCL formula in normal form is considered to satisfy the above statements.
\end{quotation}

Next, we present our results on the existence and the construction of the normal form fPCL formulas. But first, we need to note a very important observation. For this we give the following example.

\begin{example}\label{example_equiv}
	Let $P=\{ p,q \}$ be a set of ports and the \emph{fPIL} formulas $\varphi = p \ \sowedge \ ! p$ and $\varphi^\prime  = \left( p \ \sowedge \ ! p \sowedge q  \right) \sovee \left( p \ \sowedge \ ! p \sowedge \ ! q  \right) $. Those two formulas are in normal form. Considering an arbitrary De Morgan algebra we get that $\varphi\not \equiv \varphi^\prime$ since their normal forms are not equivalent. However, we prove they are equivalent over the fuzzy algebra. For this
	\begin{align*}
	\varphi^\prime  & = \left( p \ \sowedge \ ! p \sowedge q  \right) \sovee \left( p \ \sowedge \ ! p \ \sowedge \ ! q  \right) \\ & \equiv_{\textbf{F}} \left( p \ \sowedge \ ! p  \right) \sowedge \left( q   \sovee  !q \right) \\ & \equiv_{\textbf{F}} p \ \sowedge \ ! p   = \varphi 
	\end{align*}

	\noindent where the first equivalence holds since $\sowedge $ distributes over $\sovee$ and the second one by Proposition \ref{kleene_ports}. We conclude that $\varphi \equiv_{\textbf{F}} \varphi^\prime $. Analogously, we prove that $\varphi \equiv_{K_\textbf{3}} \varphi^\prime$. Also, $\varphi \equiv_{\textbf{B}} \varphi^\prime $ since both formulas are equivalent to $false$ over a Boolean algebra. However, $\varphi$ and $\varphi^\prime$ are not equivalent if we consider the four element algebra $\textbf{4}$. Let $\gamma=\{ \alpha \} \in fC(P,\textbf{4})$, where $\alpha(p) = u$ and $\alpha(q)=w$. Then $\left\Vert\varphi \right \Vert (\gamma) = u \not = 0 = \left\Vert \varphi ^\prime \right\Vert(\gamma)$. So $\varphi \not \equiv_{\textbf{4}} \varphi^\prime $. 
\end{example}

By Example \ref{example_equiv}, we observe that for the construction of the normal form of a fPCL formula we need to take into account the properties of the De Morgan algebra. In the following, we show that for every fPCL formula over $P$ and a Kleene algebra, we can effectively construct its equivalent fPCL formula in normal form.

\begin{theorem}\label{kleene_normal_form_theorem}
	Let $P$ be a set of ports and $K$ an arbitrary De Morgan algebra. Then for every \emph{fPCL} formula $\zeta_1 \in fPCL(K,P)$, $\zeta_2 \in fPCL(K_{\mathbf{3}},P)$ and $\zeta_3 \in fPCL(\mathbf{B},P)$, we can effectively construct an equivalent \emph{fPCL} formula $\zeta_1^\prime\in fPCL(K,P)$, $\zeta_2^\prime\in fPCL(K_{\mathbf{3}},P)$ and $\zeta_3^\prime\in fPCL(\mathbf{B},P)$, respectively, in normal form. The time complexity of the construction is polynomial.	
\end{theorem}

\begin{proof}
	We prove our theorem by induction on the structure of fPCL formulas over $P$ and a Kleene algbera $K_{\mathbf{3}}$. We deal with the other cases at the end of this proof. 
	
	Let $\zeta=\varphi \in fPIL(K_{\mathbf{3}}, P)$. If $\zeta$ is equal to $true$ or $false$, then we are done. Otherwise, by Propositions \ref{neg_i_oplus}, \ref{pil_true_false}, \ref{fpil_associa}, \ref{otimes_over_oplus_i} and \ref{absorpt_pil} we get its equivalent formula in $fpil$-normal form. Then we go to Step 2(2) and by Propositions \ref{otimes_over_oplus_i} and \ref{absorpt_pil} we get its equivalent formula of $\zeta$ in normal form.

	Now, let  $\zeta_1, \zeta_2$  be fPCL formulas and assume that both $\zeta_1$ and $\zeta_2$ are not equivalent to $true$ or $false$. Those cases can be treated analogously to the cases we show below and the properties of the De Morgan algebra. Consider $\zeta_1^\prime= \bigoplus_{i_1\in I_1}\biguplus_{j_1\in J_{i_1}} \varphi_{i_1,j_1}$ and $ \zeta_2^\prime = \bigoplus_{i_2\in I_2}\biguplus_{j_2\in J_{i_2}} \varphi_{i_2,j_2}$ be their equivalent normal forms, respectively. Then we go to Step 1.
	
	\begin{flushleft}
		\textbf{\underline{Step 1}}
	\end{flushleft}
	\begin{enumerate}[(1)]
		\item Firstly, let $\zeta=\zeta_1 \oplus \zeta_2$. The formula $\zeta$ is equivalent to $\zeta_1^\prime  \oplus \zeta_2^\prime$ which is of the form $\bigoplus_{i\in I} \biguplus_{j\in J_i} \varphi_{i,j} $ where $\varphi_{i,j}$ is in $fpil$-normal form for every $j\in J_i$. Then we go to Step 2. 
		
		\item Next, let $\zeta = \zeta_1 \uplus \zeta_2.$ Then 
		\begin{align*}
		\zeta &  \equiv \zeta_1^\prime \uplus \zeta_2^\prime \\ & \equiv \left( \bigoplus_{i_1\in I_1}\biguplus_{j_1\in J_{i_1}} \varphi_{i_1,j_1} \right) \uplus \left( \bigoplus_{i_2\in I_2}\biguplus_{j_2\in J_{i_2}} \varphi_{i_2,j_2}  \right) \\ & \equiv  \bigoplus_{i_1\in I_1} \bigoplus_{i_2\in I_2} \left( \biguplus_{j_1\in J_{i_1}} \varphi_{i_1,j_1}  \uplus \biguplus_{j_2\in J_{i_2}} \varphi_{i_2,j_2} \right)
		\end{align*}
		
		\noindent where the last equivalence holds by Proposition \ref{uplus_over_oplus}. Then we go to Step 2.

		\item Now, let $\zeta=\zeta_1 \otimes \zeta_2.$ Then we get 
		\begin{align*}
		\zeta  & \equiv \zeta_1^\prime \otimes \zeta_2^\prime  \\ & \equiv \left( \bigoplus_{i_1\in I_1}\biguplus_{j_1\in J_{i_1}} \varphi_{i_1,j_1} \right) \otimes \left( \bigoplus_{i_2\in I_2}\biguplus_{j_2\in J_{i_2}} \varphi_{i_2,j_2}  \right)  \\ & \equiv  \bigoplus_{(i_1, i_2) \in I_1\times I_2} \left(  \left( \biguplus_{j_1\in J_{i_1}} \varphi_{i_1,j_1} \right) \otimes \left( \biguplus_{j_2\in J_{i_2}} \varphi_{i_2,j_2} \right) \right) \\ & \equiv \bigoplus_{(i_1, i_2) \in I_1\times I_2}  \left( \left( \sim \biguplus_{j_1\in J_{i_1}} \varphi_{i_1,j_1} \uplus \biguplus_{j_2\in J_{i_2}} \varphi_{i_2,j_2}   \right) \otimes \right. \\ & \hspace*{5cm} \left.\underset{(j_1,j_2)\in J_{i_1}\times J_{i_2}}{\bigovee} \left( \varphi_{i_1,j_1} \sowedge \varphi_{i_2,j_2} \right) \right)
		\end{align*}
		
		\noindent where the third equivalence holds by Proposition \ref{otimes_over_oplus} and the fourth one by Proposition \ref{pil_coal_conj}. Consider the fPIL formula $\varphi_{(i_1, i_2)} = \underset{(j_1,j_2)\in J_{i_1}\times J_{i_2}}{\bigovee} $ $\left( \varphi_{i_1,j_1} \sowedge \varphi_{i_2,j_2} \right)$ for every $(i_1, i_2) \in I_1\times I_2$. Then by Propositions \ref{absorpt_fpcl} and \ref{otimes_distib_coale} we get 
		\begin{align*}
		\zeta    &  \equiv \bigoplus_{(i_1, i_2) \in I_1\times I_2}   \left( \biguplus_{j_1\in J_{i_1}} \left(\varphi_{i_1,j_1} \otimes \varphi_{(i_1, i_2)}\right) \uplus \biguplus_{j_2\in J_{i_2}} \left(\varphi_{i_2,j_2} \otimes \varphi_{(i_1, i_2)} \right) \uplus \right.\\ & \left. \hspace*{8cm}\left(\varphi_{(i_1, i_2)} \otimes true\right)  \right) \\ & \equiv \bigoplus_{(i_1, i_2) \in I_1\times I_2}   \left( \biguplus_{j_1\in J_{i_1}} \left(\varphi_{i_1,j_1} \sowedge \varphi_{(i_1, i_2)}\right) \uplus \biguplus_{j_2\in J_{i_2}} \left(\varphi_{i_2,j_2} \sowedge \varphi_{(i_1, i_2)} \right) \uplus \varphi_{(i_1, i_2)}  \right) .
		\end{align*}
		
		\noindent Then we go to Step 2.

		\item Let us assume that $\zeta = \neg \zeta_1$. Then 
		\begin{align*}
		\zeta &  \equiv \neg \zeta_1^\prime \\ & \equiv \neg \left( \bigoplus_{i_1\in I_1}\biguplus_{j_1\in J_{i_1}} \varphi_{i_1,j_1} \right) \\ & \equiv \bigotimes_{i_1\in I_1}\left(\neg \left(\biguplus_{j_1\in J_{i_1}} \varphi_{i_1,j_1} \right)\right) \\ & \equiv \bigotimes_{i_1\in I_1}\left( \bigoplus_{j_1\in J_{i_1}} \left(  ! \varphi_{i_1, j_1} \right) \oplus \sim \left( \bigowedge_{j_1\in J_{i_1}}  ! \varphi_{i_1. j_1}  \right) \right)
		\end{align*}
		
		\noindent where the third equivalence holds by Proposition \ref{neg} and the fourth one by Proposition \ref{pil_coal_neg}. In the sequel by applying Propositions \ref{neg_i_oplus}, \ref{otimes_over_oplus_i}, \ref{otimes_over_oplus}, \ref{otimes_distib_coale} and \ref{pil_coal_conj} we get a formula of the form $\bigoplus_{i\in I} \biguplus_{j\in J_i} \varphi_{i,j} $ where $\varphi_{i,j} = \underset{k\in K_{i,j}}{\bigovee} \varphi_{i,j,k} $ and $\varphi_{i,j,k}$ is a f-monomial for every $k\in K_{i,j}$ and $j\in J_i$.	Next, we proceed to Step 2.

	\end{enumerate}
	
	\begin{flushleft}
		\textbf{\underline{Step 2}}
	\end{flushleft}
	
	Let a formula of the form $\bigoplus_{i\in I} \biguplus_{j\in J_i} \varphi_{i,j} $ where the formula $\varphi_{i,j} = \bigovee_{k\in K_{i,j}} \varphi_{i,j,k} $ and $\varphi_{i,j,k}$ is a f-monomial for every $k\in K_{i,j}$ and $j\in J_i$. In order to get its equivalent formula in normal form we apply the following.
	\begin{enumerate}[(1)]
		\item Firstly, we apply Propositions \ref{pil_true_false}, \ref{absorpt_pil}, \ref{absorpt_fpcl} and \ref{pil_prop} (1) in order to discard any repetitions when the operations allow it. So we get a formula of the form $\bigoplus_{i^\prime\in I^\prime} \biguplus_{j^\prime\in J^\prime_{i^\prime}} \varphi_{i^\prime,j^\prime} $ where $\varphi_{i^\prime,j^\prime}= \bigoplus_{k^\prime\in K_{i^\prime,j^\prime}^\prime} \varphi_{i^\prime,j^\prime,k^\prime}$ is in $fpil$-normal form for every $(i^\prime, j^\prime) \in I^\prime\times J^\prime_{i^\prime}$.
		
		\item  Next, since we consider a Kleene algebra, we apply Proposition \ref{kleene_ports}. Let for instance a f-monomial $\varphi$ of the following form:
		\[  \varphi = \bigotimes_{p_1\in P_1} \left( p_1\otimes ! p_1 \right)\otimes \bigotimes_{p_2\in P_2} p_2 \otimes \bigotimes_{p_3\in P_3} ! p_3    \]

		\noindent where the sets $P_1, P_2, P_3 \in P$ are pairwise disjoint.   Then we consider the set $P^\prime = P\backslash (P_1\cup P_2\cup P_3) $ and by Proposition \ref{kleene_ports} we get the following: 
		\[ \varphi \equiv_{K_\mathbf{3}} \bigotimes_{p_1\in P_1} \left( p_1\otimes ! p_1 \right)\otimes \bigotimes_{p_2\in P_2} p_2 \otimes \bigotimes_{p_3\in P_3} ! p_3 \otimes \bigotimes_{p\in P^\prime } \left( p  \sovee   ! p \right).  \]
		
		\noindent We follow the above procedure for every f-monomial $\varphi$ of the form $ \bigotimes_{p_1\in P_1} \left( p_1\otimes ! p_1 \right)\otimes \bigotimes_{p_2\in P_2} p_2 \otimes \bigotimes_{p_3\in P_3} ! p_3  $. By this step we ``appear" the ports that get eliminated by the property of the Kleene algebra.

		\item Lastly, we apply Propositions \ref{pil_true_false}, \ref{absorpt_pil}, \ref{absorpt_fpcl} and \ref{pil_prop}(1) to discard again any repetitions created by the previous step.
		
	\end{enumerate}
	
	By following the steps given above, we get an equivalent formula in normal form. In order to complete our proof, we need to prove our claim for the time complexity of the algorithm presented above. In every step of our construction, we applied the distribution and idempotency properties of our logic, which are done in polynomial time. This concludes our proof. 
	
	Consider a Boolean algebra and a \emph{fPCL} formula $\zeta\in fPCL(\mathbf{B}, P)$. For the construction of its equivalent \emph{fPCL} formula over $P$ and $\mathbf{B}$ in normal form, we follow the above proof where we replace Step 2(2) with the application of Proposition \ref{boolean_prop}. If $K$ is an arbitrary De Morgan algebra and $\zeta\in fPCL(K,P)$, then for the construction of its equivalent \emph{fPCL} formula over $P$ and $K$ in normal form, we follow the above proof without Steps 2(2) and 2(3). The complexity of those constructions is again polynomial.  
\end{proof}

Next, we prove that the equivalence problem for \emph{fPCL} formulas is decidable.

\begin{theorem}\label{theor_equiv}
	Let $K$ be a De Morgan algebra and $P$ a set of ports. Then, for every $\zeta_1, \zeta_2 \in fPCL(K,P)$ the equivalence $\zeta_1\equiv \zeta_2 $ is decidable. The run time is polynomial.
\end{theorem}

\begin{proof}
	By Theorem \ref{kleene_normal_form_theorem} we can effectively construct fPCL formulas $\zeta_1^\prime , \zeta_2^\prime$ in normal form such that $\zeta_1\equiv\zeta_1^\prime$ and $\zeta_2\equiv\zeta_2^\prime$. In order to prove whether $\zeta_1$ and $\zeta_2$ are equivalent or not, we need to examine if $\zeta_1^\prime \equiv \zeta_2^\prime$ or not. For this, we write our formulas in a form of sets which we compare using Algorithm 1 given in \ref{fig:equiv}.

	Firstly, we consider that $\zeta_1^\prime=\bigoplus_{i\in I}\biguplus_{j\in J_{i}} \varphi_{i,j}$ where $\varphi_{i,j} \not = false$ is in $fpil$-normal form for every $i\in I$ and $j\in J_i$. Hence, $\zeta_1^\prime= \bigoplus_{i\in I}\biguplus_{j\in J_{i}} \bigovee_{k\in K_{i,j}} \varphi_{i,j,k}$ where $\varphi_{i,j,k}$ is a f-monomial for every $k\in K_{i,j}$. If there exists $i\in I$ and $j\in J_i$ such that $\varphi_{i,j} \equiv true$, then $\varphi_{i,j}$ can be written as $\bigovee_{k\in K_{i,j}} true$ where $K_{i,j} = \{1\}$. Analogously, $\zeta_2^\prime \equiv \bigoplus_{m\in M}\biguplus_{n\in N_m} \bigovee_{l\in L_{m,n}}  \varphi_{m,n,l}^\prime$.
	
	Next, for every f-monomial $\varphi_{i,j,k} = \bigowedge_{p\in P_{i,j,k}}p \sowedge \bigowedge_{p\in P_{i,j,k}^\prime} !p $ in $\zeta_1^\prime$, we let the set $S_{i,j,k} = \underset{p\in P_{i,j,k}}{\bigcup} \{p  \} \cup \underset{p\in P_{i,j,k}^\prime}{\bigcup} \{!p\}  $. If $\varphi_{i,j,k} =true$ then $S_{i,j,k} = \{true\}  $. Then the following set 
	\begin{enumerate}[$\bullet$]
		\item $S_{\zeta_1^\prime} =  \underset{i\in I}{\bigcup } \{ \underset{j\in J_i}{\bigcup } \{ \underset{k\in K_{i,j}}{\bigcup } \{ S_{i,j,k} \} \}    \}  $
	\end{enumerate} 
	\noindent represents $\zeta_1^\prime$ in the form of sets. If $\zeta_1^\prime =true$, then $S_{\zeta_1^\prime} = \{ \{ \{ \{  true \} \}  \} \}$ since $I = \{1\}$, $J_1 = \{1\}$, $K_{1,1} = \{1\}$ and $S_{1,1,1}=\{ true \}$. Analogously, if $\zeta_1^\prime =false$, then $S_{\zeta_1^\prime} = \{ \{ \{ \{  false \} \}  \} \}$. Next, we compute the set $S_{\zeta_2^\prime}$ which represents $\zeta_2^\prime $ in the form of sets. We need to note that the representation of a fPCL formula $\zeta = \bigoplus_{i\in I} \biguplus_{j\in J_i} \bigovee_{k\in K_{i,j}}\varphi_{i,j,k}$, which is in normal form, in the form of sets is possible since 
	\begin{enumerate}[(1)]
		\item $\biguplus_{j\in J_i} \bigovee_{k\in K_{i,j}}\varphi_{i,j,k} \not \equiv \biguplus_{j\in J_{i^\prime}} \bigovee_{k\in K_{i^\prime,j}}\varphi_{i^\prime,j,k}$ for every $i,i^\prime \in I$ with $i\not = i^\prime$, 
		\item $ \bigovee_{k\in K_{i,j}}\varphi_{i,j,k} \not \equiv \bigovee_{k\in K_{i,j^\prime}}\varphi_{i,j,k} $ for every $j, j^\prime\in J_i$ with $j\not = j^\prime$, and 
		\item $ \varphi_{i,j,k} \not \equiv \varphi_{i,j,k^\prime} $ for every $k, k^\prime \in K$ with $k\not = k^\prime.$
	\end{enumerate}
	
	In order to prove if $\zeta_1^\prime$ and $\zeta_2^\prime $ are equivalent, we need to examine if the sets $S_{\zeta_1^\prime}$ and $S_{\zeta_2^\prime}$ are equal. For this we give Algorithm 1 given in \ref{fig:equiv}. Given the sets $S_{\zeta_1^\prime}$ and $S_{\zeta_2^\prime}$ as inputs for Algorithm 1, we can decide whether $S_{\zeta_1^\prime}=S_{\zeta_2^\prime}$ or not.

	As for the time complexity of the algorithm, we prove that is polynomial. The construction of the sets $S_{\zeta_1^\prime} $ and $S_{\zeta_2^\prime} $ is done in polynomial time. As for the algorithm in \ref{fig:equiv}, we observe that there are eight nested for loops. Let that the variable of the $i$-th for loop, where $i\in \{1, \dots, 8\}$, ranges from $1$ to $n_i \in \mathbb{N}^*$. So, the number of computations in total are $n_1\cdot \ldots \cdot n_8$. Let $n=\max\{ n_1, \dots, n_8 \}$. Then $n_1\cdot \ldots \cdot n_8\leq n^8$ and the complexity is $ \mathcal{O}\left(n_1\cdot \ldots \cdot n_8\right) = \mathcal{O}\left( n^8\right).$ Hence, considering the complexity of the construction of the normal forms of $\zeta_1$ and $\zeta_2$, we conclude that the run time of the equivalence problem is polynomial. 
\end{proof}

\begin{remark}
	Let $\zeta_1, \zeta_2 \in fPCL(K,P)$. By Theorem \ref{theor_equiv} we can decide whether $\zeta_1\equiv \zeta_2$ or not. If $\zeta_1\equiv \zeta_2$ then $\zeta_1\equiv_{K_{con}} \zeta_2$ for every $K_{con}$ De Morgan algebra. However, as shown in Example \ref{example_equiv}, it is possible that $\zeta_1\equiv_{K_\mathbf{3}} \zeta_2$ but $\zeta_1\not \equiv \zeta_2$. Hence, if $\zeta_1 \not \equiv \zeta_2$, then we can examine whether $\zeta_1 \equiv_{K_\mathbf{3}} \zeta_2$ and/or $\zeta_1 \equiv_{\mathbf{B}} \zeta_2$ following the constructions in Theorems \ref{kleene_normal_form_theorem} and \ref{theor_equiv}.	
\end{remark}

\bibliographystyle{plain}
\bibliography{fPCL}

\appendix
\section{Algorithm for Decidability of Equivalence}\label{fig:equiv}

Let $\zeta_1$ and $\zeta_2$ be fPCL formulas over $P$ and $K$. By Theorem \ref{kleene_normal_form_theorem} we can effectively construct their equivalent fPCL formulas $\zeta_1^\prime$ and $\zeta_2^\prime$, respectively, in normal form. By Theorem \ref{theor_equiv}, if $S_{\zeta_1^\prime} = S_{\zeta_2^\prime}$, then $\zeta_1^\prime \equiv \zeta_2^\prime$. In order to show whether $S_{\zeta_1^\prime}$ and $S_{\zeta_2^\prime}$ are equal or not, we give Algorithm 1 where as input we have the sets $S_{\zeta_1^\prime}$ and $S_{\zeta_2^\prime}$. 

Let $P=\{p,q,r\}$ be a set of ports. For better understanding, the reader can apply the algorithm for the sets 
\begin{enumerate}[$\bullet$]
	\item $S_{\zeta_1^\prime} =  \left\{  \left\{  \left\{  \left\{ p,q\right\}, \left\{ r\right\} \right\}, \left\{ \left\{ p,!r\right\}   \right\}    \right\}, \left\{  \left\{  \left\{ !p,!r\right\}, \left\{p, r\right\} \right\}, \left\{ \left\{ p\right\}   \right\}, \left\{ \left\{ true \right\}\right\}    \right\}        \right\}   $ and 
	\item $S_{\zeta_2^\prime} =  \left\{  \left\{  \left\{  \left\{ p,q\right\}, \left\{! r\right\} \right\}, \left\{ \left\{ p\right\}   \right\}    \right\}, \left\{  \left\{  \left\{ !q\right\} \right\}, \left\{ \left\{ !p\right\}   \right\}, \left\{ \left\{ r \right\}\right\}    \right\}        \right\}   $
\end{enumerate} 
\noindent which represent the fPCL formulas 
\begin{itemize}
	\item $\zeta_1^\prime = \left(\left( \left(p\sowedge q\right) \sovee r \right) \uplus \left( p\sowedge !r \right)\right) \oplus \left( \left((!p\sowedge !r) \sovee (p\sowedge r)\right)\uplus p \uplus true \right)$ and
	\item  $\zeta_2^\prime = \left(\left( \left(p\sowedge q\right) \sovee !r \right) \uplus q\right) \oplus \left( !q\uplus !p \uplus r \right)$,
\end{itemize} 

\noindent respectively. 

\medskip

\noindent \scalebox{0.8}{\begin{minipage}[t!]{.6\textwidth}
		\fbox{\parbox{0.8\linewidth}{\begin{algorithm}[H]
					\caption{Main }\label{main:alg}
					\textbf{Input \hspace*{0.2cm}: $S_{\zeta_1^\prime}$, $S_{\zeta_2^\prime}$} \\
					\hspace*{0.5cm} 
					\begin{algorithmic}
						\If{$card(S_{\zeta_1^\prime})=card(S_{\zeta_2^\prime})$}
						\State $k\gets 0 $
						\For{$i$ in range $(1,  card(S_{\zeta_1^\prime})$}
						\For{$j$ in range $(1,  card(S_{\zeta_2^\prime})$}
						\If{SetEq$_1(S_{\zeta_1^\prime}[i], S_{\zeta_2^\prime}[j]) = true$}
						\State $k \gets k+1$
						\EndIf
						\EndFor
						\EndFor
						\If{$k=card(S_{\zeta_1^\prime})$}
						\State \text{``Equivalent"}
						\Else
						\State \text{``Not equivalent"}
						\EndIf
						\Else
						\State \text{``Not equivalent"}
						\EndIf 
					\end{algorithmic}
		\end{algorithm}}}
\end{minipage}} \hspace*{0.2cm}
\scalebox{0.8}{\begin{minipage}[t!]{.6\textwidth}
		
		\fbox{\parbox{0.8\linewidth}{\begin{algorithm}[H]
					\caption{SetEq$_1$ }
					\textbf{Input \hspace*{0.2cm}: A, B} \\
					\textbf{Output: E}
					\begin{algorithmic}
						\If{$card(A)=card(B)$}
						\State $k\gets 0 $
						\For{$i$ in range $(1,  card(A))$}
						\For{$j$ in range $(1,  card(B))$}
						\If{$SetEq_2(A[i],B[j]) = true$}
						\State $k \gets k+1$
						\EndIf
						\EndFor
						\EndFor
						\If{$k=card(A)$}
						\State $E \gets true$
						\Else
						\State $E \gets false$
						\EndIf
						\EndIf 
					\end{algorithmic}
				\end{algorithm} 
		}}
\end{minipage}}

\bigskip

\noindent \scalebox{0.8}{\begin{minipage}[t]{.6\textwidth}
		\fbox{\parbox{0.8\linewidth}{\begin{algorithm}[H]
					\caption{SetEq$_2$ }
					\textbf{Input \hspace*{0.2cm}: A, B} \\
					\textbf{Output: E}
					\begin{algorithmic}
						\If{$card(A)=card(B)$}
						\State $k\gets 0 $
						\For{$i$ in range $(1,  card(A))$}
						\For{$j$ in range $(1,  card(B))$}
						\If{$SetEq_3(A[i],B[j]) = true$}
						\State $k \gets k+1$
						\EndIf
						\EndFor
						\EndFor
						\If{$k=card(A)$}
						\State $E \gets true$
						\Else
						\State $E \gets false$
						\EndIf
						\EndIf \\
						\Return $E$
					\end{algorithmic}
		\end{algorithm}}}
\end{minipage} } \hspace*{0.2cm}
\scalebox{0.8}{\begin{minipage}[t!]{.6\textwidth}
		\fbox{\parbox{0.8\linewidth}{\begin{algorithm}[H]
					\caption{SetEq$_3$ }
					\textbf{Input \hspace*{0.2cm}: A, B} \\
					\textbf{Output: E}
					\begin{algorithmic}
						\If{$card(A)=card(B)$}
						\State $k\gets 0 $
						\For{$i$ in range $(1,  card(A))$}
						\For{$j$ in range $(1,  card(B))$}
						\If{$A[i]=B[j] $}
						\State $k \gets k+1$
						\EndIf
						\EndFor
						\EndFor
						\If{$k=card(A)$}
						\State $E \gets true$
						\Else
						\State $E \gets false$
						\EndIf
						\EndIf \\
						\Return $E$
					\end{algorithmic}
		\end{algorithm}}}
\end{minipage}}

\end{document}